\documentclass[12pt]{article}
\usepackage[margin=1in]{geometry}
\usepackage{lipsum}
\usepackage{amsmath}
\usepackage{mathrsfs}
\usepackage{multirow}
\usepackage{amsthm}
\usepackage{amssymb}
\usepackage{amsfonts}
\usepackage{graphicx}
\usepackage{longtable}
\usepackage{enumerate}
\usepackage{enumitem}
\usepackage{natbib}
\usepackage{authblk}
\usepackage{threeparttable}
\usepackage{bm}
\usepackage{url}
\usepackage{caption}
\usepackage{adjustbox}
\usepackage{indentfirst}
\usepackage{tikz}
\usepackage{pdfpages}
\usepackage{booktabs}
\usepackage{afterpage}
\usepackage{subfigure}
\newcommand{\blind}{1}

\newcommand{\R}{\mathbb{R}}
\newcommand{\norm}[1]{\left\lVert#1\right\rVert}

\newcommand{\Pn}{\mathbb{P}_n}
\newcommand{\bP}{\mathbb{P}}

\newcommand{\Gn}{\mathbb{G}_n}
\newcommand{\E}{\mathbb{E}}

\newcommand{\set}[1]{\left\{#1\right\}}
\newcommand{\setf}{\mathcal{F}}

\newcommand{\setn}{\mathcal{N}}
\newcommand{\setg}{\mathcal{G}}
\newcommand{\setp}{\mathcal{P}}

\newcommand{\parenthese}[1]{\left(#1\right)}

\newcommand{\bbrace}[1]{\left\{#1\right\}}

\newcommand{\Mn}{\mathbb{M}_n}
\newcommand{\bM}{\mathbb{M}}
\DeclareMathOperator*{\argmin}{arg\,min}
\DeclareMathOperator*{\argmax}{arg\,max}
\newcommand{\vx}{\bm{X}}
\newcommand{\vz}{\bm{Z}}
\newcommand{\vv}{\bm{V}}
\newcommand{\vb}{\bm{\beta}}
\newcommand{\ve}{\bm{\eta}}
\newcommand{\abs}[1]{\left\vert#1\right\vert}

\newtheorem{theorem}{Theorem}

\newtheorem{lemma}{Lemma}

\author[1]{Junkai Yin}

\newcommand\CoAuthorMark{\footnotemark[\arabic{footnote}]}
\author[2]{Yue Zhang\protect\CoAuthorMark}
\author[2]{Zhangsheng Yu\thanks{ yuzhangsheng@sjtu.edu.cn; Corresponding author.}}
\affil[1]{%
	Department of Statistics, Shanghai Jiao Tong University, Shanghai
	200240, PR China}
\affil[2]{%
	Department of Bioinformatics and Biostatistics, Shanghai Jiao Tong University, Shanghai
	200240, PR China}

\begin{document}

\def\spacingset#1{\renewcommand{\baselinestretch}%
{#1}\small\normalsize} \spacingset{1}

\if1\blind
{
    \title{\bf Deep partially linear transformation model for right-censored survival data}

\date{}
\maketitle
} \fi

\bigskip
\begin{abstract} 
Although the Cox proportional hazards model is well established and extensively used in the analysis of survival data, the proportional hazards (PH) assumption may not always hold in practical scenarios. The class of semiparametric transformation models extends the Cox model and also includes many other survival models as special cases. This paper introduces a deep partially linear transformation model (DPLTM) as a general and flexible regression framework for right-censored data. The proposed method is capable of avoiding the curse of dimensionality while still retaining the interpretability of some covariates of interest. We derive the overall convergence rate of the maximum likelihood estimators, the minimax lower bound of the nonparametric deep neural network (DNN) estimator, and the asymptotic normality and the semiparametric efficiency of the parametric estimator. Comprehensive simulation studies demonstrate the impressive performance of the proposed estimation procedure in terms of both the estimation accuracy and the predictive power, which is further validated by an application to a real-world dataset.
\end{abstract}

\noindent%
{\it Keywords:}  Deep learning; Minimax lower bound; Monotone splines; Partially linear transformation models; Semiparametric efficiency.
\vfill

\newpage
\spacingset{2.0}
\section{Introduction}
\label{s:intro}

The Cox proportional hazards model \citep{Cox1972} is by far one of the most common methods in survival analysis. However, it assumes proportional hazards for individuals, which may be too simplistic and often violated in practice. An example is the acquired immune deficiency syndrome (AIDS) data assembled by the U.S. Center for Disease Control, which includes 295 blood transfusion patients diagnosed with AIDS prior to July 1, 1986. One primary interest is to explore the effect of age at transfusion on the induction time, but \citet{grigoletto1999analysis} revealed that the PH assumption fails on this dataset even with the use of the reverse time PH model. The class of semiparametric transformation models emerges as a more general and flexible alternative that requires no prior assumption and has recently received tremendous attention. Most of the frequently employed survival models can be viewed as specific cases of transformation models, including the Cox proportional hazards model, the proportional odds model \citep{bennett1983analysis}, the accelerated failure time (AFT) model \citep{wei1992accelerated} and the usual Box-Cox model. Multiple estimation procedures have been thoroughly discussed for transformation models with right-censored data \citep{chen2002semiparametric}, current status data \citep{zhang2013efficient}, interval-censored data \citep{zeng2016maximum}, competing risk data \citep{fine1999analysing} and recurrent event data \citep{zeng2007semiparametric}.

Linear transformation models allow the interpretation of all covariate effects, but one limitation is that the linearity assumption is sometimes too unrealistic for complicated relationships in the real world. For instance, in the New York University Women's Health Study (NYUWHS), a question of our interest is whether the time of developing breast carcinoma is influenced by the sex hormone levels, and a strongly nonlinear relationship between them is identified by \citet{zeleniuch2004postmenopausal}. To accommodate linear and nonlinear covariate effects simultaneously, partially linear transformation models were developed \citep{ma2005penalized, lu2010estimation} and later generalized to the case with varying coefficients \citep{li2019estimation, al2022efficient}. Nevertheless, these works either only consider the simple case of univariate nonlinear effects, or assume the nonparametric effects to be additive, both of which are often inconsistent with the reality.

Public health and clinical studies in the age of big data have benefited substantially from large-scale biomedical research resources such as UK Biobank and the Surveillance, Epidemiology, and End Results (SEER) Program. Such databases often contain dozens of or even more covariates of interest to be handled simultaneously. Much important information would be left out if data from these sources are fitted by the simple linear or partially linear additive model. Recently, deep learning has rapidly evolved into a dominant and promising method in a wide range of sectors involving high-dimensional data, such as computer vision \citep{krizhevsky2012imagenet}, natural language processing \citep{collobert2011natural} and finance \citep{heaton2017deep}. Deep neural networks have also brought about significant advancements in survival analysis. They have been combined with a variety of survival models like the Cox proportional hazards model \citep{katzman2018deepsurv, Zhong2022}, the cause-specific model for competing risk data \citep{lee2018deephit}, the cure rate model \citep{xie2021promotion} and the accelerated failure time model \citep{norman2024deepaft}.

Statistical theory of deep learning associates its empirical success with its strong capability to approximate functions from specific spaces \citep{yarotsky2017error, SchmidtHieber2020}. Inspired by this, \citet{Zhong2022} considered DNNs for estimation in a partially linear Cox model, and developed a general theoretical framework to study the asymptotic properties of the partial likelihood estimators. This pioneering work has been extended to the cases of current status data \citep{wu2024deep} and interval-censored data \citep{du2024deep}. Moreover, \citet{sun2024penalized} proposed a penalized deep partially linear Cox model to simultaneously identify important features and model their effects on the survival outcome, with an application to lung cancer imaging. \citet{su2024deep} developed a DNN-based, model-free approach to estimate the conditional hazard function and carried out hypothesis tests to make inference on it. \citet{wu2023neural} and \citet{zeng2025tdcoxsnn} considered frailty and time-dependent covariates in the application of deep learning to survival analysis, respectively.

In this paper, we propose a deep partially linear transformation model for highly complex right-censored survival data. Some covariates of our primary interest are modelled linearly to keep their interpretability, while other covariate effects are approached by a deep ReLU network to alleviate the curse of dimensionality. The overall convergence rate of the estimators given by maximizing the log likelihood function is free of the nonparametric covariate dimension under proper conditions and faster than those derived using traditional smoothing methods like kernels or splines. Additionally, the parametric and nonparametric estimators are proved to be semiparametric efficient and minimax rate-optimal, respectively.

The rest of the paper is organized as follows. In Section~\ref{s:model}, we introduce the framework of our proposed method and the sieve maximum likelihood estimation procedure based on deep neural networks and monotone splines. Section~\ref{s:theory} is devoted to establishing the asymptotic properties of the estimators. In Section~\ref{s:sim}, we conduct extensive simulation studies to examine the finite sample performance of the proposed method and compare it with other models. An application to a real-world dataset is provided in Section~\ref{s:app}. Section~\ref{s:discuss} concludes the paper. Detailed proofs of lemmas and theorems, computational details, additional numerical results and further experiments are given in the Appendix.

\section{Methodology}
\label{s:model}

\subsection{Likelihood function}

We consider a study of $n$ subjects with right-censored survival data, where the survival time and the censoring time are denoted by $U$ and $C$, respectively. $\vz$ is a $p$-dimensional covariate vector impacting on the survival time linearly, and $\vx$ is a $d$-dimensional covariate vector whose effect will be modelled nonparametrically. In the presence of censoring, the observations consist of $n$ i.i.d. copies $\{\vv_i=(T_i,\Delta_i,\vz_i,\vx_i),\ i=1,\cdots,n\}$ from $\vv=(T,\Delta,\vz,\vx)$, where $T=\min\left\{U,C\right\}$ is the observed event time and $\Delta=I(U\le C)$ is the censoring indicator, with $I(\cdot)$ being the indicator function. It is generally assumed in survival analysis that $U$ is independent of $C$ conditional on $(\vz,\vx)$.

To model the effects of the covariates $(\vz,\vx)\in\R^p\times\R^d$ on the survival time $U$, the partially linear transformation models specify that
\begin{equation}
\label{eq:model}
H(U)=-\vb^{\top}\vz-g(\vx)+\epsilon,
\end{equation}
where $H$ is an unknown transformation function assumed to be strictly increasing and continuously differentiable, $\bm{\beta}\in\R^p$ denotes the unspecified parametric coefficients and $g:\R^d\rightarrow \R$ is an unknown nonparametric function. To simplify our notation, we denote the parameters to be estimated by $\ve=(\vb,H,g)$, and assume that the joint distribution of $(\Delta,\vz,\vx)$ is free of $\ve$. $\epsilon$ is an error term with a completely known continuous distribution function that is independent of $(\vz,\vx)$. 

Many useful survival models are included in the class of partially linear transformation models as special cases. For example, (\ref{eq:model}) reduces to the partially linear Cox model or the partially linear proportional odds model when $\epsilon$ follows the extreme value distribution or the standard logistic distribution, respectively. If we choose $H(t)=\log t$, (\ref{eq:model}) serves as the partially linear accelerated failure time model. When $\epsilon$ follows the normal distribution and there is no censoring, (\ref{eq:model}) generalizes the partially linear Box-Cox model.

Let ($f_{\epsilon}$, $S_{\epsilon}$, $\lambda_{\epsilon}$, $\Lambda_{\epsilon}$) and ($f_U$, $S_U$, $\lambda_U$, $\Lambda_U$) be the probability density function, survival function, hazard function and cumulative hazard function of $\epsilon$ and $U$, respectively. Then it is straightforward to verify that
\begin{align*}
&f_U(t|\vz,\vx)=H^{\prime}(t)f_{\epsilon}(H(t)+\vb^{\top}\vz+g(\vx)),\ S_U(t|\vz,\vx)=S_{\epsilon}(H(t)+\vb^{\top}\vz+g(\vx)),\\
&\lambda_U(t|\vz,\vx)=H^{\prime}(t)\lambda_{\epsilon}(H(t)+\vb^{\top}\vz+g(\vx)),\ \Lambda_U(t|\vz,\vx)=\Lambda_{\epsilon}(H(t)+\vb^{\top}\vz+g(\vx)).
\end{align*}
Therefore, the observed information of a single object under model (\ref{eq:model}) can be expressed as
\begin{align*}
\mathcal{L}(\vv)&= \left\{f_U(T|\vz,\vx)\right\}^{\Delta}\left\{S_U(T|\vz,\vx)\right\}^{1-\Delta}q(\Delta,\vz,\vx)\\
&=\left\{\lambda_U(T|\vz,\vx)\right\}^{\Delta}\exp\left\{-\Lambda_U(T|\vz,\vx)\right\}q(\Delta,\vz,\vx)\\
&=\left\{H^{\prime}(T)\lambda_{\epsilon}(H(T)+\vb^{\top}\vz+g(\vx))\right\}^{\Delta}\exp\left\{-\Lambda_{\epsilon}(H(T)+\vb^{\top}\vz+g(\vx))\right\}q(\Delta,\vz,\vx),
\end{align*}
where $q(\Delta,\vx,\vz)$ is the joint density of $(\Delta,\vx,\vz)$. Then the log likelihood function of $\ve=(\vb,H,g)$ given $\{\vv_i=(T_i,\Delta_i,\vz_i,\vx_i),\ i=1,\cdots,n\}$ can be written as
\begin{equation}
\label{eq:likelihood}
\begin{aligned}
L_n(\ve)=\sum_{i=1}^n\Big\{\Delta_i\log H^{\prime}(T_i)+\Delta_i\log\lambda_{\epsilon}(H(T_i)+&\vb^{\top}\vz_i+g(\vx_i))\\
&-\Lambda_{\epsilon}(H(T_i)+\vb^{\top}\vz_i+g(\vx_i))\Big\}.
\end{aligned}
\end{equation}

\subsection{Sieve maximum likelihood estimation}

To achieve a faster convergence rate of the maximum likelihood estimators, two different function spaces of growing capacity with respect to the sample size $n$ for the infinite-dimensional parameters $g$ and $H$ are chosen for the estimation procedure.

For the estimation of the nonparametric function $g$, we use a sparse deep ReLU network space with depth $K$, width vector $\bm{p}=(p_0,\cdots,p_{K+1})$, sparsity constraint $s$ and norm constraint $D$, which has been specified in \citet{SchmidtHieber2020} and \citet{Zhong2022} as
\begin{align*}
\mathcal{G}(K, \bm{p}, s, D)=&\Bigg\{g(\bm{x})=(W_K\sigma(\cdot)+v_K)\circ\cdots\circ (W_1\sigma(\cdot)+v_1)\circ (W_0\bm{x}+v_0):\R^{p_0}\mapsto \R^{p_{K+1}},\Big.\\
&\quad\Big. W_k\in\R^{p_{k+1}\times p_{k}},\ v_k\in\R^{p_{k+1}},\ \max\left\{\norm{W_k}_\infty,\norm{v_k}_\infty\right\}\leq 1\text{ for } k=0,\cdots, K,\Big.\\
&\quad\Big. \sum_{k=0}^{K}\parenthese{\norm{W_k}_0 + \norm{v_k}_0}\leq s,\ \norm{g}_\infty\leq D\Bigg\},
\end{align*} 
where $W_k$ and $v_k$ are the weight and bias of the $(k+1)$-th layer of the network, respectively, $\sigma(x)=\max\left\{x,0\right\}$ is the ReLU activation function operating component-wise on a vector, $\norm{\cdot}_0$ denotes the number of non-zero entries of a vector or matrix, and $\norm{\cdot}_\infty$ denotes the sup-norm of a vector, matrix or function.

To estimate the strictly increasing transformation function $H$, a monotone spline space is adopted. We assume that the support of the observed event time $T$ lies in a closed interval $[L_T,U_T]$ with $0<L_T<U_T<\tau$, where $\tau$ is the end time of the study, and partition the interval $[L_T,U_T]$ into $K_n+1$ sub-intervals with respect to the knot set 
\begin{align*}
\Upsilon=\left\{L_T=t_0<t_1<\cdots<t_{K_n+1}=U_T\right\},
\end{align*}
then we can construct $q_n=K_n+l$ B-spline basis functions $B_j(t),\ j=1,\cdots,q_n$ that are piecewise polynomials and span the space of polynomial splines $\mathcal{S}$ of order $l$ with $\Upsilon$. We set $K_n=O(n^{\nu})$ and $\underset{1\leq k\leq K_n+1}{\max}\vert t_k-t_{k-1}\vert=O(n^{-\nu})$ for some $0<\nu<1/2$ based on theoretical analysis, and $l\ge 3$ so that the spline function is at least continuously differentiable. Besides, by Theorem 5.9 of \citet{schumaker_2007}, it suffices to implement the monotone increasing restriction on the coefficients of B-spline basis functions to ensure the monotonicity of the spline function. Thus, we consider the following function space $\Psi$ which is a subset of $\mathcal{S}$:
\begin{align*}
\Psi=\left\{\sum_{j=1}^{q_n}\gamma_jB_j(t):-\infty<\gamma_1\leq\cdots\leq\gamma_{q_n}<\infty,\ t\in[L_T,U_T]\right\}.
\end{align*}
We denote the true value of $\ve=(\vb,H,g)$ by $\ve_0=(\vb_0,H_0,g_0)$, then $\ve_0$ is estimated by maximizing the log likelihood function (\ref{eq:likelihood}):
\begin{equation}
\widehat{\ve}=(\widehat{\vb},\widehat{H},\widehat{g})=\underset{(\vb,H,g)\in \R^p\times\Psi\times\mathcal{G}}{\argmax}L_n(\vb,H,g),
\label{eq:estimator}
\end{equation}
where $\mathcal{G}=\mathcal{G}(K, \bm{p}, s, \infty)$. However, it may be challenging to perform gradient-based optimization algorithms with the monotonicity constraint. We consider using a reparameterizaion approach with $\widetilde{\gamma}_1=\gamma_1$ and $\widetilde{\gamma}_j=\log(\gamma_j-\gamma_{j-1})$ for $2\le j\le q_n$ to enforce monotonicity, and then conduct optimization with respect to $\left\{\widetilde{\gamma}_j\right\}_{j=1}^{q_n}$ instead.

\section{Asymptotic properties}
\label{s:theory}

In this section, we describe the asymptotic properties of the log likelihood estimators in (\ref{eq:estimator}) under appropriate conditions. First, we impose some restrictions on the true nonparametric function $g_0$. Recall that a H\"older class of smooth functions with parameters $\alpha$, $M$ and domain $\mathbb{D}\subset\R^d$ is defined as  
\begin{align*}
\mathcal{H}_d^{\alpha}(\mathbb{D}, M)=\bbrace{g:\mathbb{D}\mapsto \R: \sum_{\bm{\kappa}:|\bm{\kappa}|<\alpha}\norm{\partial^{\bm{\kappa}} g}_{\infty}+\sum_{\bm{\kappa}:|\bm{\kappa}|=\lfloor\alpha\rfloor}\sup_{x,y\in \mathbb{D},x\ne y}\frac{|\partial^{\bm{\kappa}}g(x)-\partial^{\bm{\kappa}}g(y)|}{\norm{x-y}_{\infty}^{\alpha-\lfloor\alpha\rfloor}}\leq M},
\end{align*}
where $\partial^{\bm{\kappa}}:=\partial^{\kappa_1}\cdots\partial^{\kappa_d}$ with $\bm{\kappa}=(\kappa_1,\cdots,\kappa_d)$, and $\vert\bm{\kappa}\vert=\sum_{j=1}^d\kappa_j$. We further consider a composite smoothness function space that has been introduced in \citet{SchmidtHieber2020}:
\begin{align*}
\mathcal{H}(q,\bm{\alpha},\bm{d},\widetilde{\bm{d}}, M):=\Big\{g=&g_q\circ\cdots\circ g_0: g_i=(g_{i1},\cdots,g_{id_{i+1}})^\top\text{ and }\\
&g_{ij}\in \mathcal{H}^{\alpha_i}_{\widetilde{d}_i}([a_i,b_i]^{\widetilde{d}_i}, M) \text{, for some }|a_i|,|b_i|<M\Big\},
\end{align*}
where $\widetilde{\bm{d}}$ denotes the intrinsic dimension of the function in this space, with $\widetilde{d}_i$ being the maximal number of variables on which each of the $g_{ij}$ depends. The following composite function is an example with a relatively low intrinsic dimension:
\begin{align*}
g(x)=g_{21}\left(g_{11}\left(g_{01}\left(x_1,x_2\right),g_{02}\left(x_3,x_4\right)\right),g_{03}(x_5,x_6,x_7)\right),x\in[0,1]^7,
\end{align*}
where each $g_{ij}$ is three times continuously differentiable, then the smoothness $\bm{\alpha}=(3,3,3)$, the dimension $\bm{d}=(7,3,2,1)$ and the intrinsic dimension $\widetilde{\bm{d}}=(3,2,2)$. Furthermore, we denote $\widetilde{\alpha}_i=\alpha_i\prod_{k=i+1}^q(\alpha_k\wedge1)$ and $\delta_n=\max_{i=0,\cdots,q}n^{-\widetilde{\alpha}_i/(2\widetilde{\alpha}_i+\widetilde{d}_i)}$, and the following regularity assumptions are required to derive asymptotic properties: 
\begin{flushleft}
(C1) $K=O(\log n)$, $s=O(n\delta_n^2\log n)$ and $n\delta_n^2\lesssim \min (p_k)_{k=1,\cdots,K}\leq \max (p_k)_{k=1,\cdots,K}\lesssim n$.

(C2) The covariates $(\vz,\vx)$ take value in a bounded subset of $\R^{p+d}$ with joint probability density function bounded away from zero. Without loss of generality, we assume that the domain of $\vx$ is $[0,1]^d$. Moreover, the parameter $\vb_0$ lies in a compact subset of $\R^p$.

(C3) The nonparametric function $g_0$ lies in $\mathcal{H}_0=\{g\in\mathcal{H}(q,\bm{\alpha},\bm{d},\widetilde{\bm{d}}, M):\mathbb{E}\{g(\vx)\}=0\}$.

(C4) The $k$-th derivative of the transformation function $H_0$ is Lipschitz continuous on $[L_T,U_T]$ for any $k\ge 1$. Particularly, its first derivative is strictly positive on $[L_T,U_T]$.

(C5) The hazard function of the error term $\lambda_{\epsilon}$ is log-concave and twice continuously differentiable on $\R$. Besides, its first derivative is strictly positive on compact sets.

(C6) There is some constant $\xi>0$ such that $\mathbb{P}(\Delta=1\vert\vz,\vx)>\xi$ and $\mathbb{P}(U\ge\tau\vert\vz,\vx)>\xi$ almost surely with respect to the probability measure of $(\vz,\vx)$.

(C7) The sub-density $p(t,\bm{x},\Delta=1)$ of $(T,\vx,\Delta=1)$ is bounded away from zero and infinity on $[0,\tau]\times[0,1]^d$.

(C8) For some $k>1$, the $k$-th partial derivative of the sub-density $p(t,\bm{x},\bm{z},\Delta=1)$ of $(T,\vx,\vz,\Delta=1)$ with respect to $(t,\bm{x})$ exists and is bounded on $[0,\tau]\times[0,1]^d$. 
\end{flushleft}

Condition (C1) configures the structure of the function space $\mathcal{G}(K, \bm{p}, s, D)$ by specifying its hyperparameters which grow with the sample size. Condition (C2) is commonly used for semiparametric estimation in partially linear models. Condition (C3) yields the identifiability of the proposed model. Technical conditions (C4)-(C6) are utilized to establish the consistency and the convergence rate of the sieve maximum likelihood estimators. It is worth noting that the seemingly strong assumptions in Condition (C5) are satisfied by many familiar survival models such as the Cox proportional hazards model, the proportional odds model and the Box-Cox model. Condition (C7) guarantees the existence of the information bound for $\vb_0$. Condition (C8) establishes the asymptotic normality of $\widehat{\vb}$.

For any $\ve_1=(\vb_1,H_1,g_1)$ and $\ve_2=(\vb_2,H_2,g_2)$, define
\begin{align*}
d(\ve_1,\ve_2)=\left\{\Vert\vb_1-\vb_2\Vert^2+\Vert g_1-g_2\Vert^2_{L^2([0,1]^d)}+\Vert H_1-H_2\Vert^2_{\Psi}\right\}^{1/2}, 
\end{align*}
where $\Vert\vb_1-\vb_2\Vert^2=\sum_{i=1}^p(\beta_{i1}-\beta_{i2})^2$, $\Vert g_1-g_2\Vert^2_{L^2([0,1]^d)}=\mathbb{E}\left\{g_1(\vx)-g_2(\vx)\right\}^2$ and $\Vert H_1-H_2\Vert^2_{\Psi}=\E\left\{H_1(T)-H_2(T)\right\}^2+\E\left[\Delta\left\{H^{\prime}_1(T)-H^{\prime}_2(T)\right\}^2\right]$. With $\ve=(\vb,H,g)$ and $\vv=(T,\Delta,\vz,\vx)$, write $\phi_{\ve}(\vv)=H(T)+\vb^\top\vz+g(\vx)$, and then define
$$
\Phi_{\ve}(\vv)=\Delta\frac{\lambda_{\epsilon}^{\prime}(\phi_{\ve}(\vv))}{\lambda_{\epsilon}(\phi_{\ve}(\vv))}-\lambda_{\epsilon}(\phi_{\ve}(\vv)).
$$
Then we have the following theorems whose proofs are provided in the Appendix:

\begin{theorem}[\bf{Consistency and rate of convergence}]
\label{thm1}	
Suppose conditions (C1)–(C6) hold, and it holds that $(2w+1)^{-1}<\nu<(2w)^{-1}$ for some $w\geq 1$, then
\begin{align*}
d(\widehat{\ve},\ve_0)=O_p(\delta_n\log^2 n+n^{-w\nu}).
\end{align*}
\end{theorem}

Therefore, the proposed DNN-based method is able to mitigate the curse of dimensionality and enjoys a faster rate of convergence than traditional nonparametric smoothing methods such as kernels or splines when the intrinsic dimension $\widetilde{\bm{d}}$ is relatively low.

Furthermore, the minimax lower bound for the estimation of $g_0$ is presented below:

\begin{theorem}[\bf{Minimax lower bound}]
\label{thm2}	
Suppose conditions (C1)-(C6) hold. Define $\R^p_M=\left\{\vb\in\R^p:\Vert\vb\Vert\leq M\right\}$, then there exists a constant $0<c<\infty$, such that 
\begin{align*}
\underset{\widehat{g}}{\inf}\underset{(\vb_0,H_0,g_0)\in\R^p_M\times\Psi\times\mathcal{H}_0}{\sup}\mathbb{E}\left\{\widehat{g}(\vx)-g_0(\vx)\right\}^2\geq c\delta_n^2,
\end{align*}
where the infimum is taken over all possible estimators $\widehat{g}$ based on the observed data.
\end{theorem}

The next theorem gives the efficient score and the information bound for $\vb_0$.

\begin{theorem}[\bf{Efficient score and information bound}]
\label{thm3}
Suppose conditions (C2)-(C7) hold, then the efficient score for $\vb_0$ is 
\begin{align*}
\ell^{*}_{\vb}(\vv;\ve_0)=\left\{\vz-\bm{a}_{*}(T)-\bm{b}_*(\vx)\right\}\Phi_{\ve_0}(\vv)-\Delta\frac{\bm{a}_{*}^{\prime}(T)}{H^{\prime}_0(T)},
\end{align*}
where $(\bm{a}_{*}^{\top},\bm{b}_{*}^{\top})^\top\in\overline{\mathbb{T}}_{H_0}^p\times\overline{\mathbb{T}}_{g_0}^p$ is the least favorable direction minimizing
\begin{align*}
\mathbb{E}\left\{\left\Vert\left\{\vz-\bm{a}(T)-\bm{b}(\vx)\right\}\Phi_{\ve_0}(\vv)-\Delta\frac{\bm{a}^{\prime}(T)}{H^{\prime}_0(T)}\right\Vert^2_c\right\},
\end{align*}
with $\Vert\cdot\Vert^2_c$ denoting the component-wise square of a vector. The definitions of $\overline{\mathbb{T}}_{H_0}$ and $\overline{\mathbb{T}}_{g_0}$ are given in the Appendix. Moreover, the information bound for $\vb_0$ is
\begin{align*}
I(\bm{\beta}_0)=\mathbb{E}\left\{\ell^{*}_{\vb}(\vv;\ve_0)\right\}^{\bigotimes 2}.
\end{align*}
\end{theorem}

The last theorem states that, though the overall convergence rate is slower than $n^{-1/2}$, we can still derive the asymptotic normality of $\widehat{\vb}$ with $\sqrt{n}$-consistency.

\begin{theorem}[\bf{Asymptotic Normality}]
\label{thm4}

Suppose conditions (C1)-(C8) hold. If $(2w+1)^{-1}<\nu<(2w)^{-1}$ for some $w\geq 1$, $I(\bm{\beta}_0)$ is nonsingular and $n\delta_n^4\rightarrow 0$, then 
\begin{align*}
\sqrt n (\widehat{\bm{\beta}}-\bm{\beta}_0)=n^{-1/2}I(\vb_0)^{-1}\sum_{i=1}^n\ell^{*}_{\vb}(\vv_i;\ve_0)+o_p(1)\overset{d}{\rightarrow} N(0, I(\bm{\beta}_0)^{-1}).
\end{align*}
\end{theorem}

\section{Simulation studies}
\label{s:sim}
We carry out simulation studies in this section to investigate the finite sample performance of the proposed DPLTM method, and compare it with the linear transformation model (LTM) \citep{chen2002semiparametric} and the partially linear additive transformation model (PLATM) \citep{lu2010estimation}. Computational details are presented in the Appendix. 

In all simulations, the linearly modelled covariates $\vz$ have two independent components, where the first is generated from a Bernoulli distribution with a success probability of 0.5, and the second follows a normal distribution with both mean and variance 0.5. The covariate vector with nonlinear effects $\vx$ is 5-dimensional and generated from a Gaussian copula with correlation coefficient 0.5. Each coordinate of $\vx$ is assumed to be uniformly distributed on $[0,2]$. We take the true treatment effect $\vb_0=(1,-1)$ and consider the following three designs for the true nonparametric function $g_0(\bm{x})$ with $\bm{x}\in[0,2]^5$:
\begin{itemize}
\item \textbf{Case 1 (Linear)}: $g_0(\bm{x})=0.25(x_1+2x_2+3x_3+4x_4+5x_5-15)$,
\item \textbf{Case 2 (Additive)}: $g_0(\bm{x})=2.5\big\{\sin(2 x_1)+\cos( x_2/2)/2+\log(x_3^2+1)/3+(x_4-x_4^3)/4+(e^{x_5}-1)/5-1.27\big\}$,
\item \textbf{Case 3 (Deep)}: $g_0(\bm{x})=2.45\big\{\sin(2 x_1x_2)+\cos( x_2x_3/2)/2+\log(x_3x_4+1)/3+(x_4-x_3x_4x_5)/4+(e^{x_5}-1)/5-1.16\big\}$.
\end{itemize}
The three cases correspond to LTM, PLATM and DPLTM respectively. The intercept terms -15, -1.27 and -1.16 impose the mean-zero constraint in Condition (C4) in each case respectively, and we subtract the sample mean from the estimates to force it in practice. The factors 0.25, 2.5 and 2.45 scale the signal ratio $\text{Var}\left\{g_0(\vx)\right\}/\text{Var}\left\{\vb_0^{\top}\vz\right\}$ within $[5,7]$.

The hazard function of the error term $\epsilon$ is set to be of the form $\lambda(t)=e^t/(1+re^t)$
with $r=0,0.5,1$, i.e. the error distribution is chosen from the class of logarithmic transformations \citep{dabrowska1988estimation}. Actually, $r=0$ and $r=1$ correspond to the proportional hazards model and the proportional odds model respectively. Note that all three candidates satisfy the condition (C5) in our theoretical analysis.

The true transformation function $H_0(t)$ is set respectively as $\log t$ for $r=0$, $\log(2e^{0.5t}-2)$ for $r=0.5$ and $\log(e^t-1)$ for $r=1$. Then we can generate the survival time $U$ via its distribution function $F_U(t)=F_{\epsilon}(H_0(t)+\vb_0^{\top}\vz+g_0(\vx))$ based on the inverse transform method. The censoring time $C$ is generated from a uniform distribution on $(0, c_0)$, where the constant $c_0$ is chosen to approximately achieve the prespecified censoring rate of 40\% and 60\% ($c_0=$2.95 or 0.85 for $r=0$, $c_0=$2.75 or 0.9 for $r=0.5$, $c_0=$2.55 or 1 for $r=1$, all kept the same 
across the three different cases of the underlying function $g_0(\bm{x})$).

We conduct 200 simulation runs under each setting with sample sizes $n=1000$ or 2000. Our observations consist of $\{\vv_i=(T_i,\Delta_i,\vz_i,\vx_i),\ i=1,\cdots,n\}$, where $T_i=\min\left\{U_i,C_i\right\}$ and $\Delta_i=I(U_i\le C_i)$. We randomly split the samples into training data (80\%) and validation data (20\%). We utilize the validation data to tune the hyperparameters, and then use the training data to fit models and obtain estimates. In addition, We generated $n_{\text{test}}=200$ or 400 test samples (corresponding to $n=1000$ or 2000 respectively) that are independent of the training samples for evaluation.

To estimate the asymptotic covariance matrix $I(\vb_0)^{-1}$ for inference, where $I(\vb_0)$ is the information bound, we first estimate the least favorable directions $(\bm{a}_{*},\bm{b}_*)$ by minimizing the empirical version of the objective function given in Theorem~\ref{thm3}:
\begin{align*}
(\widehat{\bm{a}}_{*},\widehat{\bm{b}}_*)=\underset{(\bm{a},\bm{b})}{\arg\min}\ \frac{1}{n}\sum_{i=1}^n\left\Vert\left\{\vz_i-\bm{a}(T_i)-\bm{b}(\vx_i)\right\}\Phi_{\widehat{\ve}}(\vv_i)-\Delta_i\frac{\bm{a}^{\prime}(T_i)}{\widehat{H}^{\prime}(T_i)}\right\Vert^2_c.
\end{align*}
Due to the absence of closed-form expressions, we use a spline function $\sum_{j=1}^{q_n}\upsilon_jB_j(t)$ to approach $\bm{a}_*$ to achieve smoothness, and approximate $\bm{b}_{*}$ with a DNN whose input and output are $\vx$ and $\bm{b}_{*}(\vx)$, respectively. The information bound can then be estimated by
\begin{align*}
\widehat{I}(\vb_0)=\frac{1}{n}\sum_{i=1}^n\Big[\left\{\vz_i-\widehat{\bm{a}}_{*}(T_i)-\widehat{\bm{b}}_*(\vx_i)\right\}\Phi_{\widehat{\ve}}(\vv_i)-\Delta_i\frac{\widehat{\bm{a}}^{\prime}_{*}(T_i)}{\widehat{H}^{\prime}(T_i)}\Big]^{\bigotimes 2}.
\end{align*}

For evaluation of the performance of $\widehat{g}$, we compute the relative error (RE) based on the test data, which is given by
\begin{align*}
\text{RE}(\widehat{g})=\left\{\frac{\frac{1}{n_{\text{test}}}\sum_{i=1}^{n_{\text{test}}}\left[\left\{\widehat{g}(\vx_i)-\overline{\widehat{g}}\right\}-g_0(\vx_i)\right]^2}{\frac{1}{n_{\text{test}}}\sum_{i=1}^{n_{\text{test}}}\{g_0(\vx_i)\}^2}\right\}^{1/2},
\end{align*}
where $\overline{\widehat{g}}=\sum_{i=1}^{n_{\text{test}}}\widehat{g}(\vx_i)/n_{\text{test}}$.

The bias and standard deviation of the parametric estimates $\widehat{\vb}$ derived from 200 simulation runs are presented in Table~\ref{t:table1}. It is easy to see that the proposed DPLTM method provides asymptotically unbiased estimates in all situations considered. The biases for DPLTM are sometimes slightly higher than those for LTM and PLATM under Case 1, and PLATM under Case 2 respectively, which is expected because these two cases are specifically designed for the linear and additive models, respectively. However, DPLTM greatly outperforms LTM and PLATM under Case 3 with a highly nonlinear true nonparametric function $g_0$, where the other two models are remarkably more biased than DPLTM and their performance does not improve with increasing sample size. Moreover, the empirical standard deviation decreases steadily as $n$ increases for all three models under each simulation setup.

Table~\ref{t:table2} lists the empirical coverage probability of 95\% confidence intervals built with the asymptotic variance of $\widehat{\vb}$ derived from the estimated information bound $\widehat{I}(\vb_0)$. It is clear that the coverage proportion of DPLTM is generally close to the nominal level of 95\%, while PLATM gives inferior results under Case 3 and LTM shows poor coverage under both Case 2 and Case 3 because of the large bias.

Table~\ref{t:table3} reports the relative error of the norparametric estimates $\widehat{g}$ averaged over 200 simulation runs and its standard deviation on the test data. Likewise, the DPLTM estimator shows consistently strong performance in all three cases, and the metric gets smaller as the sample size increases. In contrast, LTM and PLATM behave poorly when the underlying function does not coincide with their respective model assumptions, which implies that they are unable to provide accurate estimates of complex nonparametric functions.

In the Appendix, we evaluate the accuracy in estimating the transformation function $H$ and the predictive ability of the three methods using both discrimination and calibration metrics, and compare our method with the DPLCM method proposed by \citet{Zhong2022}. We also carry out two additional simulation studies to further validate the effectiveness and robustness of the DPLTM method across various configurations.

{\linespread{1.5}
\begin{table*}
\centering
\captionsetup{justification=justified}
\caption{The bias and standard deviation of $\widehat{\vb}$ for the DPLTM, LTM and PLATM methods.}
\label{t:table1}
\scalebox{0.65}{
\begin{adjustbox}{center}
\begin{threeparttable}
\begin{tabular}{ccccccccccccccc}
     \toprule
     & & & \multicolumn{6}{c}{$\beta_1$} & \multicolumn{6}{c}{$\beta_2$} \\
     \cmidrule(r){4-9}\cmidrule(r){10-15}
     & & & \multicolumn{3}{c}{40\% censoring rate} & \multicolumn{3}{c}{60\% censoring rate} & \multicolumn{3}{c}{40\% censoring rate} & \multicolumn{3}{c}{60\% censoring rate} \\
     \cmidrule(r){4-6}\cmidrule(r){7-9}\cmidrule(r){10-12}\cmidrule(r){13-15}
     & $r$ & $n$ & DPLTM & LTM & PLATM & DPLTM & LTM & PLATM & DPLTM & LTM & PLATM & DPLTM & LTM & PLATM \\    
     \midrule
     Case 1 & 0 & 1000 & -0.0112 & 0.0212 & 0.0354 & -0.0377 & 0.0017 & 0.0209 & -0.0222 & -0.0312 & -0.0463 & -0.0107 & -0.0251 & -0.0454 \\
     (Linear) & & & (0.1023) & (0.0948) & (0.0972) & (0.1260) & (0.1109) & (0.1160) & (0.0895) & (0.0960) & (0.0982) & (0.1073) & (0.1151) & (0.1171) \\
     & & 2000 & 0.0027 & 0.0208 & 0.0263 & -0.0061 & 0.0121 & 0.0206 & -0.0167 & -0.0228 & -0.0301 & -0.0049 & -0.0131 & -0.0233 \\
     & & & (0.0680) & (0.0538) & (0.0543) & (0.0745) & (0.0691) & (0.0703) & (0.0710) & (0.0608) & (0.0617) & (0.0856) & (0.0673) & (0.0688) \\
     & 0.5 & 1000 & -0.0067 & 0.0138 & 0.0226 & -0.0210 & 0.0003 & 0.0166 & -0.0251 & -0.0333 & -0.0450 & -0.0140 & -0.0293 & -0.0470 \\
     & & & (0.1355) & (0.1168) & (0.1200) & (0.1593) & (0.1327) & (0.1362) & (0.1143) & (0.1195) & (0.1208) & (0.1337) & (0.1383) & (0.1387) \\
     & & 2000 & -0.0041 & 0.0159 & 0.0201 & -0.0011 & 0.0085 & 0.0144 & -0.0215 & -0.0216 & -0.0270 & -0.0127 & -0.0162 & -0.0243 \\
     & & & (0.0871) & (0.0681) & (0.0682) & (0.0945) & (0.0814) & (0.0829) & (0.0875) & (0.0776) & (0.0788) & (0.1008) & (0.0841) & (0.0857) \\
     & 1 & 1000 & 0.0011 & 0.0088 & 0.0185 & -0.0266 & 0.0014 & 0.0139 & -0.0208 & -0.0341 & -0.0452 & -0.0171 & -0.0334 & -0.0493 \\
     & & & (0.1576) & (0.1335) & (0.1371) & (0.1818) & (0.1527) & (0.1567) & (0.1342) & (0.1330) & (0.1342) & (0.1511) & (0.1501) & (0.1489) \\
     & & 2000 & 0.0004 & 0.0109 & 0.0169 & -0.0052 & 0.0087 & 0.0155 & -0.0195 & -0.0198 & -0.0234 & -0.0137 & -0.0200 & -0.0264 \\
     & & & (0.1007) & (0.0816) & (0.0819) & (0.1092) & (0.0903) & (0.0912) & (0.1028) & (0.0899) & (0.0914) & (0.1087) & (0.0971) & (0.0990) \\
     \\
     Case 2 & 0 & 1000 & -0.0457 & -0.3388 & -0.0353 & -0.0445 & -0.2667 & -0.0363 & 0.0380 & 0.3442 & 0.0343 & 0.0306 & 0.2717 & 0.0296 \\
     (Additive) & & & (0.0909) & (0.0866) & (0.0939) & (0.1185) & (0.1072) & (0.1071) & (0.0955) & (0.0838) & (0.0912) & (0.1167) & (0.0939) & (0.1031) \\
     & & 2000 & -0.0354 & -0.3582 & -0.0195 & -0.0350 & -0.2917 & -0.0163 & 0.0348 & 0.3552 & 0.0199 & 0.0216 & 0.2882 & 0.0159 \\
     & & & (0.0691) & (0.0581) & (0.0664) & (0.0817) & (0.0701) & (0.0730) & (0.0687) & (0.0655) & (0.0614) & (0.0841) & (0.0788) & (0.0771) \\
     & 0.5 & 1000 & -0.0373 & -0.2252 & -0.0320 & -0.0503 & -0.1929 & -0.0307 & 0.0139 & 0.2326 & 0.0283 & 0.0212 & 0.2029 & 0.0259 \\
     & & & (0.1209) & (0.1127) & (0.1167) & (0.1506) & (0.1247) & (0.1257) & (0.1232) & (0.1008) & (0.1069) & (0.1490) & (0.1098) & (0.1196) \\
     & & 2000 & -0.0343 & -0.2452 & -0.0142 & -0.0448 & -0.2157 & -0.0105 & -0.0093 & 0.2395 & 0.0194 & 0.0190 & 0.2198 & 0.0139 \\
     & & & (0.0888) & (0.0669) & (0.0775) & (0.0999) & (0.0776) & (0.0862) & (0.0902) & (0.0775) & (0.0745) & (0.1037) & (0.0895) & (0.0904) \\
     & 1 & 1000 & -0.0347 & -0.1751 & -0.0322 & -0.0520 & -0.1678 & -0.0255 & 0.0273 & 0.1820 & 0.0272 & 0.0339 & 0.1729 & 0.0281 \\
     & & & (0.1437) & (0.1300) & (0.1304) & (0.1720) & (0.1413) & (0.1454) & (0.1493) & (0.1197) & (0.1257) & (0.1636) & (0.1279) & (0.1337) \\
     & & 2000 & -0.0307 & -0.1955 & -0.0113 & -0.0401 & -0.1823 & -0.0121 & 0.0084 & 0.1869 & 0.0188 & 0.0127 & 0.1774 & 0.0164 \\
     & & & (0.1034) & (0.0771) & (0.0869) & (0.1144) & (0.0863) & (0.0942) & (0.1020) & (0.0902) & (0.0856) & (0.1159) & (0.0981) & (0.0962) \\
     \\
     Case 3 & 0 & 1000 & -0.0395 & -0.4349 & -0.2653 & -0.0474 & -0.3549 & -0.2011 & 0.0466 & 0.4310 & 0.2641 & 0.0559 & 0.3474 & 0.1990 \\
     (Deep) & & & (0.1012) & (0.0841) & (0.0849) & (0.1239) & (0.0983) & (0.1006) & (0.0982) & (0.0876) & (0.0902) & (0.1186) & (0.1033) & (0.1051) \\
     & & 2000 & -0.0322 & -0.4424 & -0.2732 & -0.0286 & -0.3672 & -0.2144 & 0.0389 & 0.4527 & 0.2867 & 0.0406 & 0.3700 & 0.2212 \\
     & & & (0.0683) & (0.0579) & (0.0614) & (0.0833) & (0.0699) & (0.0730) & (0.0720) & (0.0543) & (0.0563) & (0.0828) & (0.0669) & (0.0679) \\
     & 0.5 & 1000 & -0.0457 & -0.3267 & -0.1875 & -0.0586 & -0.2799 & -0.1483 & 0.0409 & 0.3205 & 0.1850 & 0.0382 & 0.2782 & 0.1523 \\
     & & & (0.1293) & (0.1048) & (0.1044) & (0.1577) & (0.1198) & (0.1234) & (0.1242) & (0.1097) & (0.1110) & (0.1473) & (0.1161) & (0.1173) \\
     & & 2000 & -0.0350 & -0.3347 & -0.1972 & -0.0478 & -0.2965 & -0.1698 & 0.0265 & 0.3455 & 0.2086 & 0.0244 & 0.3003 & 0.1730 \\
     & & & (0.0896) & (0.0712) & (0.0735) & (0.1022) & (0.0820) & (0.0847) & (0.0924) & (0.0681) & (0.0685) & (0.1007) & (0.0748) & (0.0851) \\
     & 1 & 1000 & -0.0570 & -0.2600 & -0.1398 & -0.0463 & -0.2444 & -0.1268 & 0.0375 & 0.2529 & 0.1411 & 0.0438 & 0.2408 & 0.1291 \\
     & & & (0.1544) & (0.1217) & (0.1226) & (0.1764) & (0.1376) & (0.1420) & (0.1450) & (0.1269) & (0.1278) & (0.1680) & (0.1304) & (0.1327) \\
     & & 2000 & -0.0344 & -0.2707 & -0.1563 & -0.0378 & -0.2592 & -0.1476 & 0.0245 & 0.2801 & 0.1666 & 0.0299 & 0.2651 & 0.1524 \\
     & & & (0.1012) & (0.0813) & (0.0831) & (0.1138) & (0.0910) & (0.0944) & (0.1028) & (0.0802) & (0.0809) & (0.1140) & (0.0863) & (0.0865) \\
     \bottomrule
\end{tabular}
\end{threeparttable}
\end{adjustbox}
}
\end{table*}

\begin{table*}
\centering
\captionsetup{justification=justified}
\caption{The empirical coverage probability of 95\% confidence intervals for $\vb_0$ for the DPLTM, LTM and PLATM methods.}
\label{t:table2}
\scalebox{0.65}{
\begin{adjustbox}{center}
\begin{threeparttable}
\begin{tabular}{ccccccccccccccc}
     \toprule
     & & & \multicolumn{6}{c}{$\beta_1$} & \multicolumn{6}{c}{$\beta_2$} \\
     \cmidrule(r){4-9}\cmidrule(r){10-15}
     & & & \multicolumn{3}{c}{40\% censoring rate} & \multicolumn{3}{c}{60\% censoring rate} & \multicolumn{3}{c}{40\% censoring rate} & \multicolumn{3}{c}{60\% censoring rate} \\
     \cmidrule(r){4-6}\cmidrule(r){7-9}\cmidrule(r){10-12}\cmidrule(r){13-15}
     & $r$ & $n$ & DPLTM & LTM & PLATM & DPLTM & LTM & PLATM & DPLTM & LTM & PLATM & DPLTM & LTM & PLATM \\    
     \midrule
     Case 1 & 0 & 1000 & 0.950 & 0.950 & 0.925 & 0.960 & 0.945 & 0.940 & 0.945 & 0.965 & 0.935 & 0.965 & 0.960 & 0.920 \\
     (Linear) & & 2000 & 0.955 & 0.930 & 0.935 & 0.950 & 0.950 & 0.935 & 0.955 & 0.960 & 0.945 & 0.950 & 0.955 & 0.930 \\
     & 0.5 & 1000 & 0.945 & 0.960 & 0.945 & 0.965 & 0.945 & 0.940 & 0.970 & 0.970 & 0.930 & 0.950 & 0.975 & 0.930 \\
     & & 2000 & 0.955 & 0.940 & 0.925 & 0.940 & 0.960 & 0.935 & 0.960 & 0.960 & 0.945 & 0.950 & 0.960 & 0.935 \\
     & 1 & 1000 & 0.950 & 0.960 & 0.935 & 0.950 & 0.960 & 0.925 & 0.945 & 0.970 & 0.930 & 0.945 & 0.970 & 0.915 \\
     & & 2000 & 0.940 & 0.935 & 0.930 & 0.960 & 0.960 & 0.950 & 0.975 & 0.955 & 0.945 & 0.945 & 0.970 & 0.930 \\
     \\
     Case 2 & 0 & 1000 & 0.935 & 0.040 & 0.940 & 0.925 & 0.030 & 0.930 & 0.950 & 0.030 & 0.935 & 0.940 & 0.315 & 0.955 \\
     (Additive) & & 2000 & 0.945 & 0.000 & 0.955 & 0.930 & 0.035 & 0.945 & 0.940 & 0.000 & 0.940 & 0.960 & 0.050 & 0.965 \\
     & 0.5 & 1000 & 0.945 & 0.445 & 0.925 & 0.930 & 0.655 & 0.920 & 0.955 & 0.420 & 0.935 & 0.945 & 0.630 & 0.935 \\
     & & 2000 & 0.930 & 0.130 & 0.945 & 0.930 & 0.310 & 0.955 & 0.945 & 0.105 & 0.930 & 0.955 & 0.335 & 0.940 \\
     & 1 & 1000 & 0.960 & 0.705 & 0.915 & 0.940 & 0.770 & 0.925 & 0.940 & 0.700 & 0.915 & 0.950 & 0.770 & 0.925 \\
     & & 2000 & 0.930 & 0.380 & 0.950 & 0.950 & 0.500 & 0.955 & 0.955 & 0.395 & 0.935 & 0.945 & 0.535 & 0.945 
     \\
     \\
     Case 3 & 0 & 1000 & 0.925 & 0.000 & 0.160 & 0.955 & 0.065 & 0.540 & 0.935 & 0.000 & 0.150 & 0.915 & 0.080 & 0.545 \\
     (Deep) & & 2000 & 0.945 & 0.000 & 0.035 & 0.920 & 0.005 & 0.205 & 0.920 & 0.000 & 0.010 & 0.935 & 0.005 & 0.135 \\
     & 0.5 & 1000 & 0.925 & 0.100 & 0.610 & 0.915 & 0.390 & 0.755 & 0.935 & 0.155 & 0.595 & 0.935 & 0.405 & 0.780 \\
     & & 2000 & 0.920 & 0.015 & 0.245 & 0.920 & 0.105 & 0.460 & 0.925 & 0.010 & 0.205 & 0.915 & 0.050 & 0.505 \\
     & 1 & 1000 & 0.930 & 0.450 & 0.785 & 0.915 & 0.575 & 0.835 & 0.955 & 0.410 & 0.800 & 0.950 & 0.565 & 0.855 \\
     & & 2000 & 0.925 & 0.140 & 0.515 & 0.925 & 0.235 & 0.625 & 0.940 & 0.105 & 0.485 & 0.955 & 0.200 & 0.650 \\
     \bottomrule
\end{tabular}

\end{threeparttable}
\end{adjustbox}
}
\end{table*}

\begin{table*}
\centering
\captionsetup{justification=justified}
\caption{The average and standard deviation of the relative error of $\widehat{g}$ for the DPLTM, LTM and PLATM methods.}
\label{t:table3}
\scalebox{0.65}{
\begin{adjustbox}{center}
\setlength{\tabcolsep}{7mm}
\begin{threeparttable}
\begin{tabular}{ccccccccc}
     \toprule
     & & & \multicolumn{3}{c}{40\% censoring rate} & \multicolumn{3}{c}{60\% censoring rate} \\
     \cmidrule(r){4-6}\cmidrule(r){7-9}
     & $r$ & $n$ & DPLTM & LTM & PLATM & DPLTM & LTM & PLATM \\    
     \midrule
     Case 1 & 0 & 1000 & 0.1302 & 0.1532 & 0.0860 & 0.1434 & 0.1001 & 0.1999 \\
     (Linear) & & & (0.0406) & (0.0357) & (0.0346) & (0.0543) & (0.0333) & (0.0421)  \\
     & & 2000 & 0.0976 & 0.0654 & 0.1037 & 0.1078 & 0.0713 & 0.1370 \\
     & & & (0.0337) & (0.0252) & (0.0226) & (0.0415) & (0.0248) & (0.0295) \\
     & 0.5 & 1000 & 0.1389 & 0.1023 & 0.1796 & 0.1557 & 0.1106 & 0.2184 \\
     & & & (0.0376) & (0.0369) & (0.0365) & (0.0477) & (0.0347) & (0.0421) \\
     & & 2000 & 0.1045 & 0.0721 & 0.1196 & 0.1172 & 0.0788 & 0.1458 \\
     & & & (0.0284) & (0.0252) & (0.0230) & (0.0340) & (0.0255) & (0.0301) \\
     & 1 & 1000 & 0.1519 & 0.1113 & 0.2001 & 0.1623 & 0.1183 & 0.2307 \\
     & & & (0.0406) & (0.0379) & (0.0377) & (0.0450) & (0.0374) & (0.0434) \\
     & & 2000 & 0.1120 & 0.0774 & 0.1319 & 0.1236 & 0.0848 & 0.1535 \\
     & & & (0.0284) & (0.0257) & (0.0240) & (0.0351) & (0.0269) & (0.0315) \\
     \\
     Case 2 & 0 & 1000 & 0.2841 & 0.7841 & 0.1532 & 0.3358 & 0.7721 & 0.1971 \\
     (Additive) & & & (0.0538) & (0.0221) & (0.0367) & (0.0741) & (0.0248) & (0.0472) \\
     & & 2000 & 0.2367 & 0.7845 & 0.1066 & 0.2617 & 0.7729 & 0.1345 \\
     & & & (0.0311) & (0.0160) & (0.0243) & (0.0476) & (0.0179) & (0.0281) \\
     & 0.5 & 1000 & 0.3223 & 0.7526 & 0.1775 & 0.3589 & 0.7592 & 0.2206 \\
     & & & (0.0444) & (0.0253) & (0.0363) & (0.0846) & (0.0267) & (0.0490) \\
     & & 2000 & 0.2618 & 0.7518 & 0.1221 & 0.2881 & 0.7575 & 0.1501 \\
     & & & (0.0336) & (0.0182) & (0.0235) & (0.0543) & (0.0193) & (0.0307) \\
     & 1 & 1000 & 0.3415 & 0.7418 & 0.1994 & 0.3652 & 0.7503 & 0.2353 \\
     & & & (0.0459) & (0.0266) & (0.0376) & (0.0782) & (0.0275) & (0.0503) \\
     & & 2000 & 0.2811 & 0.7403 & 0.1353 & 0.3079 & 0.7479 & 0.1602 \\
     & & & (0.0354) & (0.0192) & (0.0260) & (0.0597) & (0.0198) & (0.0315) \\
     \\
     Case 3 & 0 & 1000 & 0.4069 & 0.9281 & 0.7108 & 0.4287 & 0.9309 & 0.7275 \\
     (Deep) & & & (0.0549) & (0.0177) & (0.0280) & (0.0759) & (0.0186) & (0.0302) \\
     & & 2000 & 0.3421 & 0.9277 & 0.7069 & 0.3672 & 0.9301 & 0.7200 \\
     & & & (0.0416) & (0.0123) & (0.0193) & (0.0593) & (0.0133) & (0.0204) \\
     & 0.5 & 1000 & 0.4032 & 0.9214 & 0.7012 & 0.4739 & 0.9264 & 0.7217 \\
     & & & (0.0596) & (0.0199) & (0.0302) & (0.0890) & (0.0204) & (0.0314) \\
     & & 2000 & 0.3590 & 0.9203 & 0.6946 & 0.4186 & 0.9251 & 0.7110 \\
     & & & (0.0437) & (0.0140) & (0.0206) & (0.0567) & (0.0145) & (0.0212) \\
     & 1 & 1000 & 0.4516 & 0.9185 & 0.7005 & 0.4835 & 0.9234 & 0.7178 \\
     & & & (0.0624) & (0.0214) & (0.0323) & (0.0851) & (0.0216) & (0.0325) \\
     & & 2000 & 0.3788 & 0.9167 & 0.6905 & 0.4390 & 0.9217 & 0.7043 \\
     & & & (0.0487) & (0.0151) & (0.0219) & (0.0559) & (0.0151) & (0.0222) \\
     \bottomrule
\end{tabular}

\end{threeparttable}
\end{adjustbox}
}
\end{table*}

\section{Application} 
\label{s:app}

In this section, we apply the proposed DPLTM method to real-world data to demonstrate its prominent performance. We analyze lung cancer data from the Surveillance, Epidemiology, and End Results (SEER) database. We select patients who were diagnosed with lung cancer in 2015, with the age between 18 and 85 years old, the survival time longer than one month and received treatment no more than 730 days (2 years) after diagnosis. Based on previous researches \citep{anggondowati2020impact, wang2022evaluation, zhang2023prognostic}, We extract 10 important covariates, including gender, marital status, primary cancer, separate tumor nodules in ipsilateral lung, chemotherapy, age, time from diagnosis to treatment in days, CS tumor size, CS extension and CS lymph nodes. Samples with any missing covariate are discarded, which results in a dataset consisting of 28950 subjects with a censoring rate of 25.63\%. The dataset is split into a training set, a validation set and a test set with a ratio of 64:16:20. All other computational details are the same as those in simulation studies.

The main purpose of our study is to assess the predictive performance of our DPLTM method while still allowing the interpretation of some covariate effects. For the five categorial variables (gender, marital status, primary cancer, separate tumor nodules in ipsilateral lung and chemotherapy) whose effects we are mainly interested in, we denote them by $\vz$ in model (\ref{eq:model}), while the remaining five covariates are treated as $\vx$. 

The candidates for the error distribution are the same as in simulation studies, i.e. the logarithmic transformations with $r=0,0.5,1$. To obtain more accurate results, we have to select the ``optimal" one from the three transformation models. We calculate the log likelihood values on the validation data under the three fitted models for the DPLTM method, which are -6618.40, -6469.49 and -6440.13 for $r$=0, 0.5 and 1, respectively. This suggests that the model with $r=1$ (i.e. the proportional odds model) provides the best fit for this dataset and is then used for parameter estimation and prediction.

We perform a hypothesis test for each linear coefficient to explore whether the corresponding covariate has a significant effect on the survival time. Specifically, we denote the coefficient of interest by $\beta$, then the null and alternative hypotheses are $H_0:\beta=0$ and $H_1:\beta\neq0$, respectively. The test statistic is defined as $Z=\widehat{\beta}/\widehat{\sigma}$, where $\widehat{\beta}$ and $\widehat{\sigma}$ are the estimated coefficient and the estimated standard error, respectively. It can be seen from Theorem \ref{thm4} that $Z$ asymptotically follows a standard normal distribution under the null hypothesis. Thus, we can compute the asymptotic $p$-value and decide whether to reject the null hypothesis for the usual significance level $\alpha=0.05$.

Estimated coefficients (EST), estimated standard errors (ESE), test statistics and asymptotic $p$-values of the linear component for the DPLTM method with $r=1$ are given in Table~\ref{t:table5}. It is clear that all linearly modelled covariates, except the one indicating whether it is a primary cancer, are statistically significant. To be specific, females, the married, patients without separate tumor nodules in ipsilateral lung and those who received chemotherapy after diagnosis have significantly longer survival times.

In the Appendix, we also assess the predictive power of the proposed DPLTM method on this dataset with two evaluation metrics, and compare it with other models, including several machine learning models. In summary, these results reveal that our method is more effective and robust on real-world data as well.

{\linespread{1.5}
\begin{table*}
\centering
\caption{Results of the linear component for the SEER lung cancer dataset for the DPLTM method.}
\label{t:table5}
\scalebox{0.8}{
\begin{adjustbox}{center}
\begin{threeparttable}
\begin{tabular}{ccccc}
\toprule
     Covariates & EST & ESE & Test statistic & $p$-value \\    
     \midrule
     Gender (Male=1) & 0.4343 & 0.0273 & 15.9084 & $<$0.0001 \\
     Marital status (Married=1) & -0.3224 & 0.0298 & -10.8188 & $<$0.0001 \\
     Primary cancer & -0.1125 & 0.0742 & -1.5162 & 0.1295 \\
     Separate tumor nodules in ipsilateral lung & 0.4392 & 0.0330 & 13.3091 & $<$0.0001\\
     Chemotherapy & -0.4690 & 0.0309 & -15.1780 & $<$0.0001\\
     \bottomrule
\end{tabular}
\end{threeparttable}
\end{adjustbox}
}
\end{table*}
}

\section{Discussion}
\label{s:discuss}

This paper introduces a DPLTM method for right-censored survival data. It combines deep neural networks with partially linear transformation models, which encompass a number of useful models as specific cases. Our method demonstrates outstanding predictive performance while maintaining good interpretability of the parametric component. The sieve maximum likelihood estimators converge at a rate that depends only on the intrinsic dimension. We also establish the asymptotic normality and the semiparametric efficiency of the estimated coefficients, and the minimax lower bound of the deep neural network estimator. Numerical results show that DPLTM not only significantly outperforms the simple linear and additive models, but also offers major improvements over other machine learning methods.

This paper has only focused on semiparametric transformation models for right-censored survival data. It is straightforward to extend our methodology to other survival models like the cure rate model \citep{kuk1992mixture, lu2004semiparametric}, and other types of survival data such as current status data and interval-censored data. Moreover, unstructured data, such as gene sequences and histopathological images, have provided new insights into survival analysis. It is thus of great importance to combine our methodology with more advanced deep learning architectures like deep convolutional neural networks \citep{lecun1989backpropagation}, deep residual networks \citep{he2016deep} and transformers \citep{vaswani2017attention}, and develop a more general theoretical framework. Besides, a potential limitation of this study is that the sparsity constraint on the DNN is not ensured in the numerical implementation, partly because it is demanding to know certain properties of the true model (e.g. smoothness and intrinsic dimension) in practice or train a DNN with a given sparsity constraint. \citet{ohn2022nonconvex} added a clipped $L^1$ penalty to the empirical risk and showed that the sparse penalized estimator can adaptively attain minimax convergence rates for various problems. It would be beneficial to apply this technique to our methodology.

\appendix
\setcounter{table}{0}
\makeatletter
\renewcommand{\thetable}{A\arabic{table}}
\renewcommand{\thefigure}{A\arabic{figure}}

\renewcommand{\thesection}{Appendix \Alph{section}}
\renewcommand{\thesubsection}{\Alph{section}.\arabic{subsection}}

\makeatletter
\renewcommand{\@seccntformat}[1]{%
  \ifcsname the#1\endcsname
    \csname the#1\endcsname\quad 
  \fi
}
\makeatother

\section{Technical proofs}
\label{s:proofs}

\subsection{Notations}
\label{s:notation}

We denote $a_n\lesssim b_n$ as $a_n\leq Cb_n$ and $a_n\gtrsim b_n$ as $a_n\geq Cb_n$ for some constant $C>0$ and any $n\geq1$, and $a_n\asymp b_n$ implies $a_n\lesssim b_n$ and $a_n\gtrsim b_n$. For some $D>0$, we define the norm-constrained parameter spaces $\R_D^p=\set{\vb\in\R^p:\Vert\vb\Vert\leq D}$, $\setg_D=\setg(K,s,\bm{p},D)$ and 
\begin{equation*}
\Psi_D=\set{\sum_{j=1}^{q_n}\gamma_jB_j(t):-D\leq\gamma_1\leq\cdots\leq\gamma_{q_n}\leq D,\ t\in[L_T,U_T]}.
\end{equation*}
For $\ve=(\vb,H,g)$ and $\vv=(T,\Delta,\vz,\vx)$, write $\ell_{\ve}(\vv)=\Delta\log H^\prime(T)+\Delta\log\lambda_{\epsilon}(\phi_{\ve}(\vv))-\Lambda_{\epsilon}(\phi_{\ve}(\vv))$ with $\phi_{\ve}(\vv)=H(T)+\vb^\top\vz+g(\vx)$. Furthermore, we denote by $\Pn$ and $\bP$ the empirical and probability measure of $(T_i,\Delta_i,\vz_i,\vx_i)$ and $(T,\Delta,\vz,\vx)$, respectively, and let $\Gn=\sqrt{n}(\Pn-\bP)$, $\Mn(\ve)=\Pn\ell_{\ve}(\vv)=\frac{1}{n}\sum_{i=1}^n\ell_{\ve}(\vv_i)$ and $\bM(\ve)=\bP\ell_{\ve}(\vv)=\E\ell_{\ve}(\vv)$. Therefore, it is easy to see that $L_n(\ve)=n\Mn(\ve)$ and $\widehat{\ve}=\underset{\ve\in \R^p\times\Psi\times\mathcal{G}}{\argmax}L_n(\ve)=\underset{\ve\in \R^p\times\Psi\times\mathcal{G}}{\argmax}\Mn(\ve)$.

\subsection{Key lemmas and proofs}
\label{s:lemma}

\begin{lemma}
\label{lemma1}
Define $\setf=\left\{\ell_{\ve}(\vv):\ve\in\R^p_D\times\Psi_D\times\setg_D\right\}$. Suppose conditions (C1)-(C6) hold, then $\setf$ is $\bP$-Glivenko-Cantelli for any $D>0$.
\end{lemma}

\begin{proof}
Because $\R^p_D$ is a compact subset of $\R^p$, it can be covered by $\lfloor C_0(1/\varepsilon)^d\rfloor$ balls with radius $\varepsilon$, where $C_0>0$ is a constant. Hence $\log\setn(\varepsilon,\set{\vb^{\top}\vz:\vb\in\R^p_D},L^1(\bP))\lesssim d\log(1/\varepsilon)$ since $\vz$ is bounded. According to the calculation in \citet{shen1994convergence}, we have 
\begin{align*}
\log\setn(\varepsilon,\set{H(T):H\in\Psi_D},L^1(\bP))\lesssim\log\setn_{\left[\ \right]}(2\varepsilon,\set{H(T):H\in\Psi_D},L^1(\bP))\lesssim q_n\log\frac{1}{\varepsilon}.
\end{align*}
Moreover, by Theorem 4.49 of \citet{schumaker_2007}, the derivative of a spline function of order $l$ belongs to the space of polynomial splines of order $l-1$. Hence, we obtain
\begin{align*}
\log\setn(\varepsilon,\set{H^{\prime}(T):H\in\Psi_D},L^1(\bP))\lesssim\log\setn_{\left[\ \right]}(2\varepsilon,\set{H^{\prime}(T):H\in\Psi_D},L^1(\bP))\lesssim q_n\log\frac{1}{\varepsilon}.
\end{align*}
Additionally, by Lemma 6 of \citet{Zhong2022}, 
\begin{align*}
\log\setn(\varepsilon,\set{g(\vx):g\in\setg_D},L^1(\bP))\lesssim s\log\frac{L}{\varepsilon}
\end{align*} 
where $L=K\prod_{k=0}^K(p_k+1)\sum_{k=0}^Kp_kp_{k+1}$. Due to the fact that $\lambda_{\epsilon}$, $\Lambda_{\epsilon}$ and the logarithmic function are Lipschitz continuous on compact sets, the claim of the lemma follows from Lemma 9.25 in \citet{kosorok2008} and Theorem 19.13 in \citet{van2000asymptotic}.
\end{proof}

\begin{lemma}
\label{lemma2}
Suppose conditions (C2)-(C6) hold, we have
\begin{align*}
\bM(\ve)-\bM(\ve_0)\asymp-d^2(\ve,\ve_0)
\end{align*}
for all $\ve\in\set{\ve:d(\ve,\ve_0)<c_0}$ with some small $c_0>0$.
\end{lemma}

\begin{proof}
Write $\ve^*=\ve-\ve_0$ and define $\Omega(u)=\bM(\ve_0+u\ve^*)$, thus $\bM(\ve)-\bM(\ve_0)=\Omega(1)-\Omega(0)$. By Taylor expansion, there exists some $\overline{u}\in[0,1],$ such that 
\begin{equation}
\label{eq:equation1}
\Omega(1)-\Omega(0)=\Omega^{\prime}(0)+\frac12\Omega^{\prime\prime}(\overline{u}).
\end{equation}

Let $P_0$ and $P_1$ be the probability distribution of $\vv=(T,\Delta,\vz,\vx)$ with respect to $\ve_0=(\vb_0,H_0,g_0)$ and $\ve=(\vb,H,g)$, respectively, that is
\begin{align*}
P_0=&\left\{H^{\prime}_0(T)\lambda_{\epsilon}(\phi_{\ve_0}(\vv))\right\}^{\Delta}\exp\left\{-\Lambda_{\epsilon}(\phi_{\ve_0}(\vv))\right\}q(\Delta,\vz,\vx),\\
P_1=&\left\{H^{\prime}(T)\lambda_{\epsilon}(\phi_{\ve}(\vv))\right\}^{\Delta}\exp\left\{-\Lambda_{\epsilon}(\phi_{\ve}(\vv))\right\}q(\Delta,\vz,\vx).
\end{align*}
Therefore, we have $\bM(\ve)-\bM(\ve_0)=\E_{P_0}\log(P_1/P_0)=-KL(P_0,P_1)\leq 0$, where $\E_{P_0}$ is the expectation under the distribution $P_0$ and $KL(P_0,P_1)$ denotes the Kullback-Leibler distance between $P_0$ and $P_1$. This suggests that $\Omega$ attains its maximum at $u=0$, and it follows that $\Omega^{\prime}(0)=0$.
Meanwhile, direct calculation gives that
\begin{align*}
\Omega^{\prime\prime}(u)=\E\Bigg\{&-\Delta\frac{\left\{H^{\prime}(T)-H^{\prime}_0(T)\right\}^2}{\left\{H^{\prime}(u;T)\right\}^2}+\left\{\phi_{\ve}(\vv)-\phi_{\ve_0}(\vv)\right\}^2\\
&\quad\times\left[\Delta\frac{\lambda_{\epsilon}(\phi_{\ve}(u;\vv))\lambda_{\epsilon}^{\prime\prime}(\phi_{\ve}(u;\vv))-\left\{\lambda^{\prime}_{\epsilon}(\phi_{\ve}(u;\vv))\right\}^2}{\left\{\lambda_{\epsilon}(\phi_{\ve}(u;\vv))\right\}^2}-
\lambda^{\prime}_{\epsilon}(\phi_{\ve}(u;\vv))\right]\Bigg\},
\end{align*}
where $H^{\prime}(u;T)=H_0^{\prime}(T)+u\left\{H^{\prime}(T)-H^{\prime}_0(T)\right\}$ and $\phi_{\ve}(u;\vv)=\phi_{\ve_0}(\vv)+u\left\{\phi_{\ve}(\vv)-\phi_{\ve_0}(\vv)\right\}$. Conditions (C4) and (C5) imply that $H^{\prime}_0\geq C_1>0$, $\lambda^{\prime}_{\epsilon}\geq C_2>0$ and $(\log \lambda_{\epsilon})^{\prime\prime}=\{\lambda_{\epsilon}\lambda_{\epsilon}^{\prime\prime}-(\lambda_{\epsilon}^{\prime})^2\}/\lambda_{\epsilon}^{2}<0$. Consequently, it holds that 
\begin{equation}
\label{eq:equation2}
\begin{aligned}
\Omega^{\prime\prime}(\overline{u})&\lesssim-\E\left[\Delta\left\{H^{\prime}(T)-H^{\prime}_0(T)\right\}^2\right]-\E\left\{\phi_{\ve}(\vv)-\phi_{\ve_0}(\vv)\right\}^2\\
&\lesssim-\E\Big[\left\{(\vb-\vb_0)^\top \vz\right\}^2+\left\{g(\vx)-g_0(\vx)\right\}^2+\left\{H(T)-H_0(T)\right\}^2+\Delta\left\{H^{\prime}(T)-H^{\prime}_0(T)\right\}^2\Big]\\
&\lesssim-\set{\Vert\vb-\vb_0\Vert^2+\Vert g-g_0\Vert^2_{L^2([0,1]^d)}+\Vert H-H_0\Vert^2_{\Psi}}=-d^2(\ve,\ve_0),
\end{aligned}
\end{equation}
where the second inequality comes from Lemma 25.86 of \citet{van2000asymptotic}. On the other hand, by the Cauchy-Schwarz inequality, we can show that
\begin{equation}
\label{eq:equation3}
\begin{aligned}
\Omega^{\prime\prime}(\overline{u})&\gtrsim-\E\left[\Delta\left\{H^{\prime}(T)-H_0^{\prime}(T)\right\}^2\right]-\E\left\{\phi_{\ve}(\vv)-\phi_{\ve_0}(\vv)\right\}^2\\
&\gtrsim-\E\Big[\left\{(\vb-\vb_0)^\top \vz\right\}^2+\left\{g(\vx)-g_0(\vx)\right\}^2+\left\{H(T)-H_0(T)\right\}^2+\Delta\left\{H^{\prime}(T)-H^{\prime}_0(T)\right\}^2\Big]\\
&\gtrsim-\set{\Vert\vb-\vb_0\Vert^2+\Vert g-g_0\Vert^2_{L^2([0,1]^d)}+\Vert H-H_0\Vert^2_{\Psi}}=-d^2(\ve,\ve_0),
\end{aligned}
\end{equation}
Hence, combining (\ref{eq:equation1}), (\ref{eq:equation2}) and (\ref{eq:equation3}), we conclude that $\bM(\ve)-\bM(\ve_0)\asymp-d^2(\ve,\ve_0)$.
\end{proof}

\begin{lemma}
\label{lemma3}
Suppose conditions (C1)-(C6) hold. Let $\mathcal{B}_{\delta}=\set{\eta\in\R^p_D\times\Psi_D\times\setg_D:d(\ve,\ve_0)\leq\delta}$ for some $D>0$, then we have
\begin{align*}
\E^*\underset{\eta\in\mathcal{B}_{\delta}}{\sup}\abs{\Gn\left\{\ell_{\ve}(\vv)-\ell_{\ve_0}(\vv)\right\}}=O\left(\delta\sqrt{s\log\frac{L}{\delta}}+\frac{s}{\sqrt{n}}\log\frac{L}{\delta}\right), 
\end{align*}
where $\E^*$ is the outer measure and $L=K\prod_{k=0}^K(p_k+1)\sum_{k=0}^Kp_kp_{k+1}$.
\end{lemma}

\begin{proof}
Define $\setf_{\delta}=\set{\ell_{\ve}(\vv)-\ell_{\ve_0}(\vv):\ve\in\mathcal{B}_{\delta}}$ and $\Vert\Gn\Vert_{\setf_{\delta}}=\sup_{f\in\setf_{\delta}}\abs{\Gn f}=\sup_{\ve\in\mathcal{B}_{\delta}}\vert\Gn
\\ \{\ell_{\ve}(\vv)-\ell_{\ve_0}(\vv)\}\vert$. Conditions (C2), (C4) and (C5) yield
\begin{align*}
&\E\left\{\ell_{\ve}(\vv)-\ell_{\ve_0}(\vv)\right\}^2\\
\lesssim\ &\E\left[\Delta\left\{\log H^{\prime}(T)-\log H_0^{\prime}(T)\right\}^2\right]+\E\left[\Delta\left\{\log\lambda_{\epsilon}(\phi_{\ve}(\vv))-\log\lambda_{\epsilon}(\phi_{\ve_0}(\vv))\right\}^2\right]\\
&\quad+\E\left\{\Lambda_{\epsilon}(\phi_{\ve}(\vv))-\Lambda_{\epsilon}(\phi_{\ve_0}(\vv))\right\}^2 \\
\lesssim\ &\E\left[\Delta\left\{H^{\prime}(T)-H_0^{\prime}(T)\right\}^2\right]+\E\set{\phi_{\ve}(\vv)-\phi_{\ve_0}(\vv)}^2\\
\lesssim\ &\E\Big[\left\{(\vb-\vb_0)^\top \vz\right\}^2+\left\{g(\vx)-g_0(\vx)\right\}^2+\left\{H(T)-H_0(T)\right\}^2+\Delta\left\{H^{\prime}(T)-H^{\prime}_0(T)\right\}^2\Big]\\
\lesssim\ &\Vert\vb-\vb_0\Vert^2+\Vert g-g_0\Vert^2_{L^2([0,1]^d)}+\Vert H-H_0\Vert^2_{\Psi}=d^2(\ve,\ve_0).
\end{align*}
Besides, following the argument in the proof of Lemma~\ref{lemma1}, it is easy to verify that
\begin{align*}
&\log\setn_{\left[\ \right]}(\varepsilon,\set{\vb^{\top}\vz:\vb\in\R^p_D,\Vert\vb-\vb_0\Vert\leq\delta},L^2(\bP))\lesssim d\log\frac{\delta}{\varepsilon},\\
&\log\setn_{\left[\ \right]}(\varepsilon,\set{g(\vx):g\in\setg_D,\Vert g-g_0\Vert_{L^2([0,1]^d)}\leq\delta},L^2(\bP))\lesssim s\log\frac{L}{\varepsilon},\\
&\log\setn_{\left[\ \right]}(\varepsilon,\set{H(T):H\in\Psi_D,\Vert H-H_0\Vert_{\Psi}\leq\delta},L^2(\bP))\lesssim q_n\log\frac{\delta}{\varepsilon},\\
&\log\setn_{\left[\ \right]}(\varepsilon,\set{H^{\prime}(T):H\in\Psi_D,\Vert H-H_0\Vert_{\Psi}\leq\delta},L^2(\bP))\lesssim q_n\log\frac{\delta}{\varepsilon}.
\end{align*}
Thus, with $d\leq s$, $q_n\leq s$ and $\delta\leq L$, we can get
\begin{align*}
\log\setn_{\left[\ \right]}(\varepsilon,\setf_{\delta},L^2(\bP))\lesssim d\log\frac{\delta}{\varepsilon}+2q_n\log\frac{\delta}{\varepsilon}+s\log\frac{L}{\varepsilon}\lesssim s\log\frac{L}{\varepsilon}.
\end{align*}
Consequently, we can derive the bracketing integral of $\setf_{\delta}$, 
\begin{align*}
J_{\left[\ \right]}(\varepsilon,\mathcal{F}_{\delta},L^2(\bP))& =\int_0^\delta\sqrt{1+\log\mathcal{N}_{\left[\ \right]}(\varepsilon,\mathcal{F}_\delta,L^2(\bP))}d\varepsilon   \\
&\lesssim\int_0^\delta\sqrt{1+s\log\frac L\varepsilon}d\varepsilon  \\
&\begin{aligned}&=\frac{2L}{s}e^{\frac
1s}\int_{\sqrt{1+s\log\frac{L}{\delta}}}^{\infty}y^2e^{-\frac{y^2}{s}}dy\end{aligned} \\
&\asymp\delta\sqrt{s\log\frac L\delta}.
\end{align*}
This, in conjunction with Lemma 3.4.2 in \citet{van1996weak}, leads to
\begin{align*}
\mathbb{E}^*\|\mathbb{G}_n\|_{\mathcal{F}_\delta}& \lesssim J_{\left[\ \right]}(\varepsilon,\mathcal{F}_\delta,L^2(\bP))\left\{1+\frac{J_{\left[\ \right]}(\varepsilon,\mathcal{F}_\delta,L^2(\bP))}{\delta^2\sqrt{n}}\right\}  \\
&\lesssim\delta\sqrt{s\log\frac L\delta}+\frac s{\sqrt{n}}\log\frac L\delta,
\end{align*}
which completes the proof.\\
\end{proof}

\subsection{Proof of Theorem 1}
We consider the following norm-constrained estimator:
\begin{align}
\label{eq:equation4}
\widehat{\ve}_D=(\widehat{\vb}_D,\widehat{H}_D,\widehat{g}_D)=\underset{(\vb,H,g)\in \R^p_D\times\Psi_D\times\mathcal{G}_D}{\argmax}\Mn(\vb,H,g).
\end{align}
It is easy to see that $\bP\set{d(\widehat{\ve},\ve_0)<\infty}=1$ since $\widehat{\ve}$ maximizes $\Mn(\ve)$, thus it suffices to show that $d(\widehat{\ve}_D,\ve_0)=O_p(\delta_n\log^2 n+n^{-w\nu})$ for some sufficiently large constant $D$. 

First, we show that $d(\widehat{\ve}_D,\ve_0)\overset{p}{\rightarrow}0$ by applying Theorem 5.7 of \citet{van2000asymptotic}. It follows directly from Lemma~\ref{lemma1} that
\begin{align}
\label{eq:equation5}
\underset{\ve\in\R^p_D\times\Psi_D\times\setg_D}{\sup}\abs{\Mn(\ve)-\bM(\eta)}\overset{p}{\rightarrow}0,
\end{align}
and Lemma~\ref{lemma2} indicates that
\begin{align}
\label{eq:equation6}
\underset{d(\ve,\ve_0)\geq c_0}{\sup}\bM(\ve)<\bM(\ve_0)
\end{align}
for some small constant $c_0>0$. Furthermore, we define
\begin{align}
\label{eq:equation7}
\widetilde{g}=\underset{g\in\setg(K,s,\bm{p},D)}{\argmin}\left\Vert g-g_0\right\Vert_{L^2([0,1]^d)}.
\end{align}
By the proof of Theorem 1 in \citet{SchmidtHieber2020}, we have $\Vert\widetilde{g}-g_0\Vert_{L^2([0,1]^d)}=O_p(\delta_n\log^2n)$. Besides, Lemma A1 of \citet{lu2007estimation} implies that there exists some $\widetilde{h}\in\Psi_D^{(1)}=\{H^{\prime}:H\in\Psi_D\}$, such that
\begin{align}
\label{eq:equation8}
\Vert\widetilde{h}-H^{\prime}_0\Vert_{\infty}=O_p(n^{-w\nu}). 
\end{align}
We then define 
\begin{align}
\label{eq:equation9}
\widetilde{H}(t)=H_0(L_T)+\int_{L_T}^{t}\widetilde{h}(s)ds,\ L_T\leq t\leq U_T,
\end{align}
and now we can use $\widetilde{H}^{\prime}$ in place of $\widetilde{h}$ in the subsequent parts of the proof. It is clear that
\begin{align}
\label{eq:equation10}
\Vert\widetilde{H}-H_0\Vert_{\infty}=\underset{t\in[L_T,U_T]}{\sup}\left\vert\int_{L_T}^{t}\left\{\widetilde{H}^{\prime}(s)-H^{\prime}_0(s)\right\}ds\right\vert=O_p(n^{-w\nu}).
\end{align}
(\ref{eq:equation8}) and (\ref{eq:equation10}) further give that 
\begin{equation}
\begin{aligned}
\label{eq:equation11}
\Vert \widetilde{H}-H_0\Vert_{\Psi}=\E\left[\left\{\widetilde{H}(T)-H_0(T)\right\}^2+\Delta\left\{\widetilde{H}^{\prime}(T)-H^{\prime}_0(T)\right\}^2\right]^{1/2}=O_p(n^{-w\nu}).
\end{aligned}
\end{equation}
Thus, combining (\ref{eq:equation5}), Lemma 2 and the law of large numbers, we obtain 
\begin{equation}
\label{eq:equation12}
\begin{aligned}
&\big\vert\Mn(\vb_0,\widetilde{H},\widetilde{g})-\Mn(\vb_0,H_0,g_0)\big\vert\\
\leq\ &\big\vert\Mn(\vb_0,\widetilde{H},\widetilde{g})-\bM(\vb_0,\widetilde{H},\widetilde{g})\big\vert+\big\vert\bM(\vb_0,\widetilde{H},\widetilde{g})-\bM(\vb_0,H_0,g_0)\big\vert\\
&+\big\vert\bM(\vb_0,H_0,g_0)-\Mn(\vb_0,H_0,g_0)\big\vert\\
=\ &o_p(1).
\end{aligned}
\end{equation}
By the definition of $\widehat{\ve}_D=(\widehat{\vb}_D,\widehat{H}_D,\widehat{g}_D)$, we get
\begin{align}
\label{eq:equation13}
\Mn(\widehat{\vb}_D,\widehat{H}_D,\widehat{g}_D)\geq\Mn(\vb_0,\widetilde{H},\widetilde{g})=\Mn(\vb_0,H_0,g_0)-o_p(1).
\end{align}
Hence, we prove the consistency by verifying the conditions with (\ref{eq:equation5}), (\ref{eq:equation6}) and (\ref{eq:equation13}).

Next, we employ Theorem 3.4.2 of \citet{van1996weak} to derive that $d(\widehat{\ve},\ve_0)=O_p(\delta_n\log^2 n+n^{-w\nu}).$ Define $\mathcal{A}_{\delta}=\left\{\ve\in\R^p_D\times\Psi_D\times\setg_D:\delta/2\leq d(\ve,\ve_0)\leq\delta\right\}$, Lemma~\ref{lemma2} yields that
\begin{align}
\label{eq:equation14}
\underset{\ve\in\mathcal{A}_{\delta}}{\sup}\left\{\bM(\ve)-\bM(\ve_0)\right\}\lesssim-\delta^2.
\end{align}
Define $\varphi_n(\delta)=\delta\sqrt{s\log\frac{L}{\delta}}+\frac{s}{\sqrt{n}}\log\frac{L}{\delta}+\sqrt{n}(\delta_n\log^2n+n^{-w\nu})^2$ and $\theta_n=\delta_n\log^2n+n^{-w\nu}$. It follows from Lemma~\ref{lemma3} that
\begin{align}
\label{eq:equation15}
\E^*\underset{\ve\in\mathcal{A}_{\delta}}{\sup}\sqrt{n}\set{(\Mn-\bM)(\ve)-(\Mn-\bM)(\ve_0)}\lesssim\varphi_n(\delta).
\end{align}
Moreover, condition (C1) leads to 
\begin{align}
\label{eq:equation16}
\theta_n^{-2}\varphi_n(\theta_n)\leq\sqrt{n}.
\end{align}
With $\widetilde{g}$ and $\widetilde{H}$ defined in (\ref{eq:equation7}) and (\ref{eq:equation9}) respectively, by analogy to (\ref{eq:equation12}), it holds that
\begin{equation}
\label{eq:equation17}
\begin{aligned}
&\big\vert\Mn(\vb_0,\widetilde{H},\widetilde{g})-\Mn(\vb_0,H_0,g_0)\big\vert\\
\leq\ &\big\vert(\Mn-\bM)(\vb_0,\widetilde{H},\widetilde{g})-(\Mn-\bM)(\vb_0,H_0,g_0)\big\vert+\big\vert\bM(\vb_0,\widetilde{H},\widetilde{g})-\bM(\vb_0,H_0,g_0)\big\vert
\\
\lesssim\ &O_p(n^{-1/2}\varphi_n(\theta_n))+\Vert\widetilde{H}-H_0\Vert^2_{\Psi}+\Vert\widetilde{g}-g_0\Vert^2_{L^2([0,1]^d)}\\
\lesssim\ &O_p(\theta_n^2).
\end{aligned}
\end{equation}
Since $\widehat{\ve}_D=(\widehat{\vb}_D,\widehat{H}_D,\widehat{g}_D)$ is the norm-constrained maximizer of the log likelihood function,
\begin{align}
\label{eq:equation18}
\Mn(\widehat{\vb}_D,\widehat{H}_D,\widehat{g}_D)\geq\Mn(\vb_0,\widetilde{H},\widetilde{g})=\Mn(\vb_0,H_0,g_0)-O_p(\theta_n^2).
\end{align}
Consequently, combining (\ref{eq:equation14}), (\ref{eq:equation15}), (\ref{eq:equation16}) and (\ref{eq:equation18}), we have 
\begin{align*}
d(\widehat{\ve}_D,\ve_0)=O_p(\delta_n\log^2 n+n^{-w\nu}).
\end{align*}
and it follows that $d(\widehat{\ve},\ve_0)=O_p(\delta_n\log^2 n+n^{-w\nu})$. Therefore, the proof is completed.

\subsection{Proof of Theorem 2}
Let $P_{(\vb_0,H_0,g_0)}$ be the probability distribution with respect to the parameter $\vb_0$, the transformation function $H_0$ and the nonparametric smooth function $g_0$. Then we define
\begin{align*}
\setp_0&=\{P_{(\vb_0,H_0,g_0)}:\vb_0\in\R_M^p,H_0\in\Psi\text{ and }g_0\in\mathcal{H}_0\},\\
\setp_1&=\{P_{(\vb_0,H_0,g_0)}:\vb_0\in\R_M^p,H_0\in\Psi_1\text{ and }g_0\in\mathcal{H}_1\},
\end{align*}
where $M>0$ is a constant, $\Psi_1=\left\{\sum_{j=1}^{q_n}\gamma_jB_j(t):0=\gamma_1\leq\cdots\leq\gamma_{q_n}<\infty,\ t\in[L_T,U_T]\right\}$, and $\mathcal{H}_1=\mathcal{H}(q,\bm{\alpha},\bm{d},\widetilde{\bm{d}}, M/2)$.

For any $(\vb,H_1,g_1)\in\R^p_M\times\Psi_1\times\mathcal{H}_1$, it is easy to see that $P_{(\vb,H_1,g_1)}\overset{d}{=}P_{(\vb,H_1+c^{\prime},g_1-c^{\prime})}$ with $c^{\prime}=\E\set{g_1(\vx)}$. Note that $\sum_{j=1}^{q_n}B_j(t)\equiv 1$ by Theorem 4.20 of \citet{schumaker_2007}, it follows that $H_1+c^{\prime}$ is an element of  $\left\{\sum_{j=1}^{q_n}\gamma_jB_j(t):c^{\prime}=\gamma_1\leq\cdots\leq\gamma_{q_n}<\infty,\ t\in[L_T,U_T]\right\}$, which is a subset of $\Psi$. Thus $P_{(\vb,H_1+c^{\prime},g_1-c^{\prime})}\in\setp_0$, which further implies that $\setp_1$ is a subset of $\setp_0$.

Suppose that $\widehat{g}_1$ is an estimator of $g_1\in\mathcal{H}_1$ from the observations $\{\vv_i=(T_i,\Delta_i,\vz_i,\vx_i),\ i=1,\cdots,n\}$ under some model $P_{(\vb,H_1,g_1)}\in\setp_1$, then $\widehat{g}_0:=\widehat{g}_1-c^{\prime}$ with $c^{\prime}=\E\set{g_1(\vx)}$ is also an estimator of $g_0:=g_1-c^{\prime}$ based on the same observations under $P_{(\vb,H_1+c^{\prime},g_1-c^{\prime})}\in\setp_0$. By the fact that $\widehat{g}_1-g_1=\widehat{g}_0-g_0$, we have
\begin{equation}
\label{eq:equation19}
\begin{aligned}
&\inf_{\widehat{g}_0}\sup_{(\vb_0,H_0,g_0)\in\mathbb{R}_M^p\times\Psi\times\mathcal{H}_0}\mathbb{E}_{P_{(\vb_0,H_0,g_0)}}\{\widehat{g}_0(\vx)-g_0(\vx)\}^2\\
\geq&\inf_{\widehat{g}_1}\sup_{(\vb_1,H_1,g_1)\in\mathbb{R}_M^p\times \Psi_1\times\mathcal{H}_1}\mathbb{E}_{P_{(\vb_1,H_1,g_1)}}\{\widehat{g}_1(\vx)-g_1(\vx)\}^2.
\end{aligned}
\end{equation}
Therefore, it suffices to find a lower bound for the right hand side of (\ref{eq:equation19}) to obtain that for the left hand side of (\ref{eq:equation19}).

Let $(\vb_0,H_0)\in\R^p_M\times\Psi_1$ and $g^{(0)},g^{(1)}\in\mathcal{H}_1$, we denote by $P_0$ and $P_1$ the joint distribution of $\{\vv_i=(T_i,\Delta_i,\vz_i,\vx_i),\ i=1,\cdots,n\}$ under $P_{(\vb_0,H_0,g^{(0)})}$ and $P_{(\vb_0,H_0,g^{(1)})}$, respectively. By analogy to the proof of Lemma 2, there exists constants $a_1,a_2>0$, such that
\begin{equation}
\label{eq:equation20}
\begin{aligned}
KL(P_1,P_0)&\leq a_1d^2_{P_1}\left\{(\vb_0,H_0,g^{(1)}),(\vb_0,H_0,g^{(0)})\right\}\\
&=a_1\sum_{i=1}^n\E_{P_1}\left\{g^{(1)}(\vx_i)-g^{(0)}(\vx_i)\right\}^2\leq a_2n\Vert g^{(1)}-g^{(0)}\Vert^2_{L^2([0,1]^d)},
\end{aligned}
\end{equation}
where
\begin{align*}
d^2_{P_1}(\ve_1,\ve_2)=\sum_{i=1}^n\E_{P_1}\big[&\left\{(\vb_1-\vb_2)^\top \vz_i\right\}^2+\left\{g_1(\vx_i)-g_2(\vx_i)\right\}^2+\left\{H_1(T_i)-H_2(T_i)\right\}^2\\
&+\Delta\left\{H^{\prime}_1(T_i)-H^{\prime}_2(T_i)\right\}^2\big]
\end{align*}
for any $\ve_1=(\vb_1,H_1,g_1)$ and $\ve_2=(\vb_2,H_2,g_2)$. According to the proof of Theorem 3 in \citet{SchmidtHieber2020}, there exist $g^{(0)},\cdots,g^{(N)}\in\mathcal{H}_1$ and constants $b_1,b_2>0$, such that
\begin{equation}
\label{eq:equation21}
\begin{aligned}
&\Vert g^{(k)}-g^{(l)}\Vert_{L^2([0,1]^d)}\geq2b_1\delta_n>0\text{ for any }1\leq k,l\leq N\\
&\text{ and  }\quad\frac{a_2n}N\sum_{k=1}^N\|g^{(k)}-g^{(0)}\|_{L^2([0,1]^d)}^2\leq b_2\log N.
\end{aligned}
\end{equation}
Therefore, combining (\ref{eq:equation20}) and (\ref{eq:equation21}), by Theorem 2.5 of \citet{tsybakov2009nonparametric}, we can show that
\begin{align*}
\inf_{\widehat{g}_1}\sup_{g_1\in\mathcal{H}_1}\mathbb{P}(\|\widehat{g}_1-g_1\|_{L^2([0,1]^d)}\geq b_1\delta_n)\geq\frac{\sqrt{N}}{1+\sqrt{N}}\left(1-2b_2-\sqrt{\frac{2b_2}{\log N}}\right),
\end{align*}
which gives that 
\begin{align*}
\underset{\widehat{g}_1}{\inf}\underset{(\vb_1,H_1,g_1)\in\R^p_M\times\Psi_1\times\mathcal{H}_1}{\sup}\mathbb{E}_{P_{(\vb_1,H_1,g_1)}}\left\{\widehat{g}_1(\vx)-g_1(\vx)\right\}^2\geq c\delta_n^2,
\end{align*}
for some constant $c>0$. This completes the proof.

\subsection{Proof of Theorem 3}
We first describe the function spaces $\overline{\mathbb{T}}_{H_0}$ and $\overline{\mathbb{T}}_{g_0}$. Let $\Psi_{H_0}$ be the collection of all subfamilies $\left\{H_{s_1}\in L^2([L_T,U_T])\cap C^1([L_T,U_T]):H_{s_1}\text{ is strictly increasing, } s_1\in(-1,1)\right\}$ such that $\lim_{s_1\rightarrow 0}\Vert s_1^{-1}(H_{s_1}-H_0)-a\Vert_{L^2([L_T,U_T])}=0$, where $a\in L^2([L_T,U_T])\cap C^1([L_T,U_T])$, and then define
\begin{align*}
\mathbb{T}_{H_0}=\Big\{a\in L^2([L_T,U_T])\cap C^1([L_T,U_T]):&\lim_{s_1\rightarrow 0}\Vert s_1^{-1}(H_{s_1}-H_0)-a\Vert_{L^2([L_T,U_T])}=0\\
&\text{ for some subfamily }
\left\{H_{s_1}:s_1\in(-1,1)\right\}\in\Psi_{H_0}\Big\},
\end{align*}
Similarly, let $\mathcal{H}_{g_0}$ denote the collection of all subfamilies $\left\{g_{s_2}\in L^2([0,1]^d):s_2\in(-1,1)\right\}\subset\mathcal{H}_0$ such that $\lim_{s_2\rightarrow 0}\Vert s_2^{-1}(g_{s_2}-g_0)-b\Vert_{L^2([0,1]^d)}=0$ with $b\in L^2([0,1]^d)$, and then define
\begin{align*}
\mathbb{T}_{g_0}=\Big\{b\in L^2([0,1]^d):\lim_{s_2\rightarrow 0}\Vert s_2^{-1}(g_{s_2}-g_0)-&b\Vert_{L^2([0,1]^d)}=0\\
&\text{ for some subfamily }\left\{g_{s_2}:s_2\in(-1,1)\right\}\in\mathcal{H}_{g_0}\Big\}.
\end{align*}
Let $\overline{\mathbb{T}}_{H_0}$ and $\overline{\mathbb{T}}_{g_0}$ be the closed linear spans of $\mathbb{T}_{H_0}$ and $\mathbb{T}_{g_0}$, respectively.

We consider a parametric submodel $\{(\vb,H_{s_1},g_{s_2}):s_1,s_2\in(-1,1)\}$,  where $\{H_{s_1}:s_1\in(-1,1)\}\in\Psi_{H_0}$, $H_{s_1}\vert_{s_1=0}=H_0$ and $\{g_{s_2}:s_2\in(-1,1)\}\in\mathcal{H}_{g_0}$, $g_{s_2}\vert_{s_2=0}=g_0$. By definitions of the subfamilies $\Psi_{H_0}$ and $\mathcal{H}_{g_0}$, there exist $a\in\overline{T}_{H_0}$ and $b\in\overline{T}_{g_0}$ such that
\begin{align*}
\frac{\partial H_{s_1}}{\partial {s_1}}\bigg\vert_{s_1=0}=a,\quad\frac{\partial  H^{\prime}_{s_1}}{\partial s_1}\bigg\vert_{s_1=0}=a^{\prime}\ \text{ and }\ \frac{\partial g_{s_2}}{\partial s_2}\bigg\vert_{s_2=0}=b.
\end{align*}
Thus, by differentiating the log likelihood function with respect to $\vb$, $s_1$ and $s_2$ at $\vb=\vb_0$, $s_1=0$ and $s_2=0$, we get the score function for $\vb_0$ and the score operators for $H_0$ and $g_0$, which are respectively defined as
\begin{align*}
&\dot{\ell}_{\vb}(\vv;\ve_0)=\frac{\partial}{\partial\vb}\ell_{(\vb,H_0,g_0)}(\vv)\bigg\vert_{\vb=\vb_0}=\vz\Phi_{\ve_0}(\vv),\\
&\dot{\ell}_{H}(\vv;\ve_0)[a]=\frac{\partial}{\partial s_1}\ell_{(\vb_0,H_{s_1},g_0)}(\vv)\bigg\vert_{s_1=0}=a(T)\Phi_{\ve_0}(\vv)+\Delta \frac{a^{\prime}(T)}{H^{\prime}(T)},\\
&\dot{\ell}_{g}(\vv;\ve_0)[b]=\frac{\partial}{\partial s_2}\ell_{(\vb_0,H_0,g_{s_2})}(\vv)\bigg\vert_{s_2=0}=b(\vx)\Phi_{\ve_0}(\vv).
\end{align*}
By chapter 3 of \citet{kosorok2008}, the efficient score function for $\vb_0$ is given by 
\begin{align*}
\ell^*_{\vb}(\vv;\ve_0)=\dot{\ell}_{\vb}(\vv;\ve_0)-\Pi_{H_0,g_0}[\dot{\ell}_{\vb}(\vv;\ve_0)\vert\dot{\mathbf{P}}_1+\dot{\mathbf{P}}_2]
\end{align*}
where $\Pi_{H_0,g_0}[\dot{\ell}_{\vb}(\vv;\ve_0)\vert\dot{\mathbf{P}}_1+\dot{\mathbf{P}}_2]$ is the projection of $\dot{\ell}_{\vb}(\vv;\ve_0)$ onto the sumspace $\dot{\mathbf{P}}_1+\dot{\mathbf{P}}_2$, with $\dot{\mathbf{P}}_1=\{\dot{\ell}_{H}(\vv;\ve_0)[a]:a\in\overline{\mathbb{T}}_{H_0}\}$ and $\dot{\mathbf{P}}_2=\{\dot{\ell}_{g}(\vv;\ve_0)[b]:b\in\overline{\mathbb{T}}_{g_0}\}$. Furthermore, $\Pi_{H_0,g_0}[\dot{\ell}_{\vb}(\vv;\ve_0)\vert\dot{\mathbf{P}}_1+\dot{\mathbf{P}}_2]$ can be obtained by deriving the least favorable direction $(\bm{a}_{*}^{\top},\bm{b}_{*}^{\top})^\top\in\overline{\mathbb{T}}_{H_0}^p\times\overline{\mathbb{T}}_{g_0}^p$, which satisfies   
\begin{align*}
&\E\bigg[\Big\{\dot{\ell}_{\vb}(\vv;\ve_0)-\dot{\ell}_{H}(\vv;\ve_0)[\bm{a}_{*}]-\dot{\ell}_{g}(\vv;\ve_0)[\bm{b}_{*}]\Big\}\dot{\ell}_{H}(\vv;\ve_0)[a]\bigg]=0,\text{ for all }a\in\overline{\mathbb{T}}_{H_0},\\
&\E\bigg[\Big\{\dot{\ell}_{\vb}(\vv;\ve_0)-\dot{\ell}_{H}(\vv;\ve_0)[\bm{a}_{*}]-\dot{\ell}_{g}(\vv;\ve_0)[\bm{b}_{*}]\Big\}\dot{\ell}_{g}(\vv;\ve_0)[b]\bigg]=0,\text{ for all }b\in\overline{\mathbb{T}}_{g_0}.
\end{align*}
This leads to the conclusion that $(\bm{a}_{*}^{\top},\bm{b}_{*}^{\top})^\top$ is the minimizer of
\begin{align*}
&\E\left\{\left\Vert\dot{\ell}_{\vb}(\vv;\ve_0)-\dot{\ell}_{H}(\vv;\ve_0)[\bm{a}]-\dot{\ell}_{g}(\vv;\ve_0)[\bm{b}]\right\Vert_c^2\right\}\\
=\ &\E\left\{\left\Vert\left\{\vz-\bm{a}(T)-\bm{b}(\vx)\right\}\Phi_{\ve_0}(\vv)-\Delta\frac{\bm{a}^{\prime}(T)}{H^{\prime}_0(T)}\right\Vert^2_c\right\}.
\end{align*}
By conditions (C2)-(C7), Lemma 1 of \citet{stone1985additive}, and Appendix A.4 in \citet{bickel1993efficient}, the minimizer $(\bm{a}_{*}^{\top},\bm{b}_{*}^{\top})^\top$ is well defined. Hence, the efficient score is
\begin{align*}
\ell^{*}_{\vb}(\vv;\ve_0)&=\dot{\ell}_{\vb}(\vv;\ve_0)-\dot{\ell}_{H}(\vv;\ve_0)[\bm{a}_{*}]-\dot{\ell}_{g}(\vv;\ve_0)[\bm{b}_{*}]\\
&=\left\{\vz-\bm{a}_{*}(T)-\bm{b}_*(\vx)\right\}\Phi_{\ve_0}(\vv)-\Delta\frac{\bm{a}_{*}^{\prime}(T)}{H^{\prime}(T)},    
\end{align*}
and the information matrix is 
\begin{align*}
I(\vb_0)=\mathbb{E}\left\{\ell^{*}_{\vb}(\vv;\ve_0)\right\}^{\bigotimes 2}.
\end{align*}

\subsection{Proof of Theorem 4}
Using the mean value theorem and the Cauchy-Schwarz inequality, we have
\begin{align*}
&\bP\left\{\ell^{*}_{\vb}(\vv;\widehat{\ve})-\ell^{*}_{\vb}(\vv;\ve_0)\right\}^2\\
=\ &\bP\left\{\ell^{*}_{\vb}(\vv;\ve_0+\rho(\widehat{\ve}-\ve_0))\big\vert_{\rho=1}-\ell^{*}_{\vb}(\vv;\ve_0+\rho(\widehat{\ve}-\ve_0))\big\vert_{\rho=0}\right\}^2\\
=\ &\bP\left\{\frac{d}{d\rho}\ell^{*}_{\vb}(\vv;\ve_0+\rho(\widehat{\ve}-\ve_0))\bigg\vert_{\rho=\overline{\rho}}\right\}^2\\
=\ &\bP\Bigg\{\frac{d}{d\left[\left\{\vb_0+\rho(\widehat{\vb}-\vb_0)\right\}^{\top}\vz\right]}\ell^{*}_{\vb}(\vv;\ve_0+\rho(\widehat{\ve}-\ve_0))\bigg\vert_{\rho=\overline{\rho}}\left\{(\widehat{\vb}-\vb_0)^\top\vz\right\}\\
&\quad+\frac{d}{d\left[g_0(\vx)+\rho\left\{\widehat{g}(\vx)-g_0(\vx)\right\}\right]}\ell^{*}_{\vb}(\vv;\ve_0+\rho(\widehat{\ve}-\ve_0))\bigg\vert_{\rho=\overline{\rho}}\left\{\widehat{g}(\vx)-g_0(\vx)\right\}\\
&\quad+\frac{d}{d\left[H_0(T)+\rho\left\{\widehat{H}(T)-H_0(T)\right\}\right]}\ell^{*}_{\vb}(\vv;\ve_0+\rho(\widehat{\ve}-\ve_0))\bigg\vert_{\rho=\overline{\rho}}\left\{\widehat{H}(T)-H_0(T)\right\}\\
&\quad+\frac{d}{d\left[H^{\prime}_0(T)+\rho\left\{\widehat{H}^{\prime}(T)-H^{\prime}_0(T)\right\}\right]}\ell^{*}_{\vb}(\vv;\ve_0+\rho(\widehat{\ve}-\ve_0))\bigg\vert_{\rho=\overline{\rho}}\left\{\widehat{H}^{\prime}(T)-H^{\prime}_0(T)\right\}\Bigg\}^2\\
\lesssim\ &\bP\Big[\left\{(\widehat{\vb}-\vb_0)^\top \vz\right\}^2+\left\{\widehat{g}(\vx)-g_0(\vx)\right\}^2+\left\{\widehat{H}(T)-H_0(T)\right\}^2+\Delta\left\{\widehat{H}^{\prime}(T)-H^{\prime}_0(T)\right\}^2\Big]\\
\lesssim\ &\Vert\widehat{\vb}-\vb_0\Vert^2+\Vert \widehat{g}-g_0\Vert^2_{L^2([0,1]^d)}+\Vert \widehat{H}-H_0\Vert^2_{\Psi}=d^2(\widehat{\ve},\ve_0)\overset{p}{\rightarrow}0,
\end{align*}
where $\overline{\rho}\in[0,1]$. Since $\lambda_{\epsilon},\Lambda_{\epsilon}$ and the logarithmic function are Lipschitz continuous on compact sets, with conditions (C2), (C4) and (C5), it follows from Theorem 2.10.6 of \citet{van1996weak} that $\{\ell^{*}_{\vb}(\vv;\ve):d(\ve,\ve_0)\leq\delta\}$ is a $\bP$-Donsker class, and $\ell^{*}_{\vb}(\vv;\widehat{\ve})$ belongs to this class for sufficiently large $n$ as a consequence of Theorem 1. Then Theorem 19.24 of \citet{van2000asymptotic} yields
\begin{align}
\label{eq:equation22}
(\Pn-\bP)\left\{\ell^{*}_{\vb}(\vv;\widehat{\ve})-\ell^{*}_{\vb}(\vv;\ve_0)\right\}=o_p(n^{-1/2}).
\end{align}
For any $\bm{a}\in\Psi^p$ and $\bm{b}\in\mathcal{G}^p$, define the function
\begin{align*}
\Gamma(\bm{\mu};\vv)=\Pn\left[\Delta\log \left\{\widehat{H}^\prime(T)-\bm{\mu}^{\top}\bm{a}^{\prime}(T)\right\}+\Delta\log\lambda_{\epsilon}(\zeta_{\widehat{\ve}}(\bm{\mu};\vv))-\Lambda_{\epsilon}(\zeta_{\widehat{\ve}}(\bm{\mu};\vv))\right],
\end{align*}
where $\zeta_{\widehat{\ve}}(\bm{\mu};\vv)=\left\{\widehat{H}(T)-\bm{\mu}^{\top}\bm{a}(T)\right\}+(\widehat{\vb}+\bm{\mu})^{\top}\vz+\left\{\widehat{g}(\vx)-\bm{\mu}^{\top}\bm{b}(\vx)\right\}$. By differentiating $\Gamma$ at $\bm{\mu}=0$ and the definition of $\widehat{\ve}$, we get 
\begin{align*}
\Pn\left\{\dot{\ell}_{\vb}(\vv;\widehat{\ve})-\dot{\ell}_{H}(\vv;\widehat{\ve})\left[\bm{a}\right]-\dot{\ell}_{g}(\vv;\widehat{\ve})\left[\bm{b}\right]\right\}=0.
\end{align*}
From \citet{lu2007estimation}, there exists $\bm{a}_n=(a_{n,1},\cdots,a_{n,p})^{\top}\in\Psi^p$ such that $\Vert a_{*,m}-a_{n,m}\Vert_{\infty}=O(n^{-w\nu})$ and $\Vert a^{\prime}_{*,m}-a^{\prime}_{n,m}\Vert_{\infty}=O(n^{-w\nu})$, $1\leq m\leq p$, thus $\Vert a_{*,m}-a_{n,m}\Vert_{\Psi}=O(n^{-w\nu})$. Note that $\bP \dot{\ell}_H(\vv;\ve_0)[a_{*,m}-a_{n,m}]=0$ because of Lemma \ref{lemma2}, we can write $\Pn\dot{\ell}_H(\vv;\widehat{\ve})[a_{*,m}-a_{n,m}]=J_{n,m}^{(1)}+J_{n,m}^{(2)}$, where $J_{n,m}^{(1)}=(\Pn-\bP)\{\dot{\ell}_H(\vv;\widehat{\ve})[a_{*,m}-a_{n,m}]\}$ and $J_{n,m}^{(2)}=\bP\{\dot{\ell}_H(\vv;\widehat{\ve})[a_{*,m}-a_{n,m}]-\dot{\ell}_H(\vv;\ve_0)[a_{*,m}-a_{n,m}]\}$. By analogy to the proof of (\ref{eq:equation22}), we can show that $J_{n,m}^{(1)}=o_p(n^{-1/2})$ and $J_{n,m}^{(2)}\le[\bP\{\dot{\ell}_H(\vv;\widehat{\ve})[a_{*,m}-a_{n,m}]-\dot{\ell}_H(\vv;\ve_0)[a_{*,m}-a_{n,m}]\}^2]^{1/2}\lesssim d(\widehat{\ve},\ve_0)\Vert a_{*,m}-a_{n,m}\Vert_{\Psi}=o_p(n^{-1/2})$, $1\leq m\leq p$ under conditions $(2w+1)^{-1}<\nu<(2w)^{-1}$ for some $w\geq 1$ and $n\delta_n^4\rightarrow 0$, which implies that
\begin{align*}
\Pn\dot{\ell}_H(\vv;\widehat{\ve})[\bm{a}_{*}-\bm{a}_n]=o_p(n^{-1/2}).
\end{align*}
From \citet{SchmidtHieber2020}, there exists $\bm{b}_n=(b_{n,1},\cdots,b_{n,p})^{\top}\in\mathcal{G}^p$ such that $\Vert b_{*,m}-b_{n,m}\Vert_{L^2([0,1]^d)}=O(\delta_n\log^2n)$, $1\leq m\leq p$. Similarly, we have
\begin{align*}
\Pn\dot{\ell}_g(\vv;\widehat{\ve})[\bm{b}_{*}-\bm{b}_n]=o_p(n^{-1/2}).
\end{align*}
Then it holds that
\begin{equation}
\begin{aligned}
\label{eq:equation23}
&\ \Pn\left\{\ell^{*}_{\vb}(\vv;\widehat{\ve})\right\}=\Pn\left\{\dot{\ell}_{\vb}(\vv;\widehat{\ve})-\dot{\ell}_{H}(\vv;\widehat{\ve})\left[\bm{a}_*\right]-\dot{\ell}_{g}(\vv;\widehat{\ve})\left[\bm{b}_*\right]\right\}\\
=&\ \Pn\left[\dot{\ell}_{\vb}(\vv;\widehat{\ve})-\left\{\dot{\ell}_{H}(\vv;\widehat{\ve})[\bm{a}_n]+\dot{\ell}_H(\vv;\widehat{\ve})[\bm{a}_{*}-\bm{a}_n]\right\}-\left\{\dot{\ell}_{g}(\vv;\widehat{\ve})[\bm{b}_n]+\dot{\ell}_g(\vv;\widehat{\ve})[\bm{b}_{*}-\bm{b}_n]\right\}\right]\\
=&\ o_p(n^{-1/2}).
\end{aligned}
\end{equation}
Additionally, the Taylor expansion gives that
\begin{equation*}
\begin{aligned}
\bP&\left\{\ell^*_{\vb}(\vv;\widehat{\ve})-\ell^*_{\vb}(\vv;\ve_0)\right\}=-\bP\left\{\ell^*_{\vb}(\vv;\ve_0)\dot{\ell}_{\vb}(\vv;\ve_0)^{\top}(\widehat{\vb}-\vb_0)\right\}\\
&\quad\quad-\bP\left[\ell^*_{\vb}(\vv;\ve_0)\left\{\dot{\ell}_{H}(\vv;\ve_0)[\widehat{H}-H_0]+\dot{\ell}_{g}(\vv;\ve_0)[\widehat{g}-g_0]\right\}\right]+O_p(d^2(\widehat{\ve},\ve_0)).
\end{aligned}
\end{equation*}
According to the proof of Theorem 3, we know that the efficient score $\ell^*_{\vb}(\vv;\ve_0)$ is orthogonal to $\dot{\mathbf{P}}_1+\dot{\mathbf{P}}_2$, which is the tangent sumspace generated by the scores $\dot{\ell}_{H}(\vv;\ve_0)[a]$ and $\dot{\ell}_{g}(\vv;\ve_0)[b]$. We then obtain that 
\begin{equation}
\label{eq:equation24}
\begin{aligned}
\bP\left\{\ell^*_{\vb}(\vv;\widehat{\ve})-\ell^*_{\vb}(\vv;\ve_0)\right\}&=-\bP\left\{\ell^*_{\vb}(\vv;\ve_0)\dot{\ell}_{\vb}(\vv;\ve_0)^{\top}(\widehat{\vb}-\vb_0)\right\}+O_p(d^2(\widehat{\ve},\ve_0))\\
&=-\bP\left\{\ell^*_{\vb}(\vv;\ve_0)\ell^*_{\vb}(\vv;\ve_0)^{\top}(\widehat{\vb}-\vb_0)\right\}+O_p(d^2(\widehat{\ve},\ve_0))\\
&=-I(\vb_0)(\widehat{\vb}-\vb_0)+o_p(n^{-1/2})
\end{aligned}
\end{equation}
with $(2w+1)^{-1}<\nu<(2w)^{-1}$ for some $w\geq 1$ and $n\delta_n^4\rightarrow 0$. Hence, combining (\ref{eq:equation22}), (\ref{eq:equation23}) and (\ref{eq:equation24}), we conclude by the central limit theorem that
\begin{align*}
\sqrt{n}(\widehat{\vb}-\vb_0)&=\sqrt{n}I(\vb_0)^{-1}\left\{I(\vb_0)(\widehat{\vb}-\vb_0)\right\}\\
&=\sqrt{n}I(\vb_0)^{-1}\left[-\bP\left\{\ell^*_{\vb}(\vv;\widehat{\ve})-\ell^*_{\vb}(\vv;\ve_0)\right\}+o_p(n^{-1/2})\right]\\
&=\sqrt{n}I(\vb_0)^{-1}\left[-\Pn\left\{\ell^*_{\vb}(\vv;\widehat{\ve})-\ell^*_{\vb}(\vv;\ve_0)\right\}+o_p(n^{-1/2})\right]\\
&=\sqrt{n}I(\vb_0)^{-1}\left[\Pn\left\{\ell^*_{\vb}(\vv;\ve_0)\right\}+o_p(n^{-1/2})\right]\\
&=n^{-1/2}I(\vb_0)^{-1}\sum_{i=1}^n\ell^{*}_{\vb}(\vv_i;\ve_0)+o_p(1)\overset{d}{\rightarrow} N(0, I(\vb_0)^{-1}).
\end{align*}
Therefore, the proof is completed.

\section{Computational details}
\label{s:details}
Here we provide some computational details for the numerical experiments. The DPLTM method is implemented by PyTorch \citep{NEURIPS2019_9015}. The model is fitted by maximizing the log likelihood function with respect to the parameters $\vb$, $\widetilde{\gamma}_j$'s, $W_k$'s and $v_k$'s, all contained in one framework and simultaneously updated through the back-propagation algorithm in each epoch. The Adam optimizer \citep{adam} is employed due to its efficiency and reliability. All components of $\vb$ and all $\widetilde{\gamma}_j$'s are initialized to 0 and -1, respectively, while PyTorch’s default random initialization algorithm is applied to $W_k$'s and $v_k$'s.

The hyperparameters, including the number of hidden layers, the number of neurons in each hidden layer, the number of epochs, the learning rate \citep{goodfellow2016deep}, the dropout rate \citep{srivastava2014dropout} and the number of B-spline basis functions are tuned based on the log likelihood on the validation data via a grid search. We set the number of neurons in each hidden layer to be the same for convenience. We evenly partition the support set $[L_T,U_T]$ and use cubic splines (i.e. $l$=4) to estimate $H$ to achieve sufficient smoothness, with the number of interior knots $K_n$ chosen in the range of $\lfloor n^{1/3}\rfloor$ to $2\lfloor n^{1/3}\rfloor$, and then the number of basis functions $q_n=K_n+l$ can be determined. Candidates for other hyperparameters are summarized in Table~\ref{t:webtable1}. It is worth noting that the optimal combination of hyperparameters can vary from case to case (e.g., different error distributions or censoring rates) and thus should be selected out separately under each setting. 

\begin{table}
\centering
\caption{Candidate values of hyperparameters.}
\label{t:webtable1}
\begin{tabular}{cc}
\toprule
Hyperparameter & Candidate set \\
\midrule
Number of layers & $\left\{\text{1, 2, 3, 4, 5}\right\}$ \\
\midrule
Number of layers & $\left\{\text{5, 10, 15, 20, 50}\right\}$ \\
\midrule
Number of epochs & $\left\{\text{100, 200, 500}\right\}$ \\
\midrule
Learning rate & $\left\{\text{1e-3, 2e-3, 5e-3, 1e-2}\right\}$ \\
\midrule
Dropout rate & $\left\{\text{0, 0.1, 0.2, 0.3}\right\}$ \\
\bottomrule
\end{tabular}
\end{table}

To avoid overfitting, we use the strategy of early stopping \citep{goodfellow2016deep}. To be specific, if the validation loss (i.e. the negative log likelihood on the validation data) stops decreasing for a predetermined number of consecutive epochs, which is an indication of overfitting, we then terminate the training process and obtain the estimates.

For the estimation of the information bound, a cubic spline function is employed to approach $\bm{a}_*$ with the same number of basis functions as in the estimation of $H$, and the DNN utilized to approximate $\bm{b}_*$ has 2 hidden layers with 10 neurons in each. The number of epochs, the learning rate and the dropout rate used to minimize the objective function are 100, 2e-3 and 0, respectively. Therefore, the computational burden is relatively mild. Specifically, the time spent estimating the asymptotic variances is roughly 4 seconds in each simulation run when the sample size $n=1000$, and is approximately doubled when $n$ increases to 2000.

\section{Additional numerical results}
\label{s:furres}
\subsection{Results on the transformation function}
\label{s:H}
Better estimation of the transformation function $H$ brings on more reliable prediction of the survival probability. To measure the estimation accuracy of $\widehat{H}$, we compute the weighted integrated squared error (WISE) defined as
\begin{align*}
\text{WISE}(\widehat{H})=\frac{1}{T_{\text{max}}}\int_0^{T_{\text{max}}}\left\{\widehat{H}(t)-H_0(t)\right\}^2dt,
\end{align*}
where $T_{\text{max}}=\underset{1\le i\le n}{\max}\ T_{i}$ is the maximum  observed event time. Because the interval over which we take the integral varies from case to case, we introduce the weight function $w(t)=1/T_{\text{max}}$ to conveniently compare the results across various configurations. In practice, the integration is carried out numerically using the trapezoidal rule.

Table~\ref{t:webtable2} demonstrates the performance in estimating $H$, where we display the weighted integrated squared error averaged over 200 simulation runs along with its standard deviation. DPLTM leads to only marginally larger WISE than LTM under Case 1 and PLATM under Case 1 and Case 2, but produces considerably more accurate results than the two methods under the more complex setting of Case 3. It can also be observed that low censoring rates generally yield better estimates when the simulation setting meets the model assumption.

\begin{table}
\centering
\caption{The average and standard deviation of the weighted integrated squared error of $\widehat{H}(t)$ for the DPLTM, LTM and PLATM methods.}
\label{t:webtable2}
\scalebox{0.65}{
\begin{adjustbox}{center}
\setlength{\tabcolsep}{7mm}
\begin{tabular}{ccccccccc}
     \toprule
     
     & & & \multicolumn{3}{c}{40\% censoring rate} & \multicolumn{3}{c}{60\% censoring rate} \\
     \cmidrule(r){4-6}\cmidrule(r){7-9}
     & $r$ & $n$ & DPLTM & LTM & PLATM & DPLTM & LTM & PLATM \\    
     \midrule
     Case 1 & 0 & 1000 & 0.0266 & 0.0180 & 0.0209 & 0.0271 & 0.0201 & 0.0216 \\
     (Linear) & & & (0.0213) & (0.0141) & (0.0154) & (0.0195) & (0.0165) & (0.0143) \\
     & & 2000 & 0.0164 & 0.0054 & 0.0102 & 0.0205 & 0.0129 & 0.0157 \\
     & & & (0.0106) & (0.0063) & (0.0069) & (0.0122) & (0.0070) & (0.0083) \\
     & 0.5 & 1000 & 0.0362 & 0.0256 & 0.0279 & 0.0408 & 0.0252 & 0.0289 \\
     & & & (0.0233) & (0.0164) & (0.0185) & (0.0257) & (0.0172) & (0.0156) \\
     & & 2000 & 0.0210 & 0.0116 & 0.0130 & 0.0231 & 0.0125 & 0.0127 \\
     & & & (0.0167) & (0.0084) & (0.0086) & (0.0151) & (0.0105) & (0.0105) \\
     & 1 & 1000 & 0.0488 & 0.0244 & 0.0276 & 0.0511 & 0.0284 & 0.0316 \\
     & & & (0.0355) & (0.0167) & (0.0164) & (0.0327) & (0.0193) & (0.0188) \\
     & & 2000 & 0.0307 & 0.0158 & 0.0145 & 0.0253 & 0.0137 & 0.0148 \\
     & & & (0.0238) & (0.0114) & (0.0107) & (0.0186) & (0.0122) & (0.0128) \\
     \\
     Case 2 & 0 & 1000 & 0.0334 & 0.1321 & 0.0203 & 0.0373 & 0.1333 & 0.0272 \\
     (Additive) & & & (0.0187) & (0.0381) & (0.0151) & (0.0215) & (0.0547) & (0.0190) \\
     & & 2000 & 0.0239 & 0.1288 & 0.0102 & 0.0255 & 0.1369 & 0.0190 \\
     & & & (0.0096) & (0.0239) & (0.0072) & (0.0146) & (0.0394) & (0.0114) \\
     & 0.5 & 1000 & 0.0329 & 0.1158 & 0.0282 & 0.0356 & 0.1013 & 0.0331 \\
     & & & (0.0189) & (0.0484) & (0.0173) & (0.0217) & (0.0533) & (0.0200) \\
     & & 2000 & 0.0228 & 0.1097 & 0.0135 & 0.0255 & 0.1016 & 0.0149 \\
     & & & (0.0147) & (0.0295) & (0.0094) & (0.0171) & (0.0382) & (0.0113) \\
     & 1 & 1000 & 0.0502 & 0.1128 & 0.0351 & 0.0547 & 0.0828 & 0.0366 \\
     & & & (0.0279) & (0.0526) & (0.0220) & (0.0341) & (0.0488) & (0.0265) \\
     & & 2000 & 0.0329 & 0.1016 & 0.0178 & 0.0364 & 0.783 & 0.0173 \\
     & & & (0.0186) & (0.0301) & (0.0142) & (0.0199) & (0.0321) & (0.0136) \\
     \\
     Case 3 & 0 & 1000 & 0.0508 & 0.1890 & 0.0868 & 0.0542 & 0.2260 & 0.0979 \\
     (Deep) & & & (0.0328) & (0.0284) & (0.0235) & (0.0335) & (0.0710) & (0.0524) \\
     & & 2000 & 0.0356 & 0.1920 & 0.0902 & 0.0362 & 0.2203 & 0.0942 \\
     & & & (0.0190) & (0.0215) & (0.0194) & (0.0216) & (0.0433) & (0.0335) \\
     & 0.5 & 1000 & 0.0501 & 0.1974 & 0.0831 & 0.0576 & 0.1827 & 0.0785 \\
     & & & (0.0378) & (0.0429) & (0.0319) & (0.0447) & (0.0720) & (0.0508) \\
     & & 2000 & 0.0382 & 0.2010 & 0.0839 & 0.0364 & 0.1768 & 0.0745 \\
     & & & (0.0245) & (0.0322) & (0.0252) & (0.0301) & (0.0435) & (0.0318) \\
     & 1 & 1000 & 0.0558 & 0.2021 & 0.0865 & 0.0578 & 0.1472 & 0.0755 \\
     & & & (0.0392) & (0.0590) & (0.0395) & (0.0434) & (0.0653) & (0.0459) \\
     & & 2000 & 0.0375 & 0.2004 & 0.0829 & 0.0459 & 0.1388 & 0.0689 \\
     & & & (0.0267) & (0.0408) & (0.0323) & (0.0291) & (0.0380) & (0.0294) \\
     \bottomrule
\end{tabular}
\end{adjustbox}
}
\end{table}

\subsection{Results on prediction}
\label{s:ICI}
We utilize both discrimination and calibration metrics to assess the predictive performance of the three methods. Discrimination means the ability to distinguish subjects with the event of interest from those without, while calibration refers to the agreement between observed and estimated probabilities of the outcome.

The discrimination metric we adopt is the concordance index (C-index) by \citet{harrell1982evaluating}. The C-index is one of the most commonly used metrics to evaluate the predictive power of models in survival analysis. It measures the probability that the predicted survival times preserve the ranks of true survival times, which is defined as
\begin{align*}
\text{C}=\mathbb{P}(\widehat{T}_i<\widehat{T}_j\vert T_i<T_j,\Delta_i=1),
\end{align*}
where $\widehat{T}_i$ denotes the predicted survival time of the $i$-th individual. Larger C-index values indicate better predictive performance. For the semiparametric transformation model, the C-index can be empirically calculated as
\begin{align*}
\widehat{\text{C}}=\frac{\sum_{i=1}^{n_{\text{test}}}\sum_{j=1}^{n_{\text{test}}}\Delta_i1(T_i\leq T_j)1(\widehat{\vb}\vz_i+\widehat{g}(\vx_i)\geq\widehat{\vb}\vz_j+\widehat{g}(\vx_j))}{\sum_{i=1}^{n_{\text{test}}}\sum_{j=1}^{n_{\text{test}}}\Delta_i1(T_i\leq T_j)}.
\end{align*}

The calibration metric we choose is the integrated calibration index (ICI) by \citet{austin2020graphical}. It quantifies the consistency between observed and estimated probabilities of the time-to-event outcome prior to a specified time $t_0$. It is given by
\begin{align*}
\text{ICI}(t_0)=\frac{1}{n_{\text{test}}}\sum_{i=1}^{n_{\text{test}}}\left\vert \widetilde{P}_i^{t_0}-\widehat{P}_i^{t_0}\right\vert,
\end{align*}
where $\widehat{P}_i^{t_0}=F_{\epsilon}(\widehat{H}(t_0)+\widehat{\vb}^{\top}\vz_i+\widehat{g}(\vx_i))$ is the predicted probability of the outcome prior to $t_0$ for the $i$-th individual, and $\widetilde{P}_i^{t_0}$ is an estimate of the observed probability given the predicted probability. Specifically, we fit the hazard regression model \citep{kooperberg1995hazard}:
\begin{align*}
\log(h(t))=\psi(\log(-\log(1-\widehat{P}^{t_0})),t),  
\end{align*}
where $h(t)$ is the hazard function of the outcome and $\psi$ is a nonparametric function to be estimated. Then $\widetilde{P}_i^{t_0}=1-\exp\left\{-\int_0^{t_0} \widehat{h}_i(s)ds\right\}$, with $\widehat{h}_i(t)=\exp\left\{\widehat{\psi}(\log(-\log(1-\widehat{P}_i^{t_0})),t)\right\}$. Smaller ICI values imply greater predictive ability. In practice, we compute the ICI at the 25th ($t_{25}$), 50th ($t_{50}$) and 75th ($t_{75}$) percentiles of observed event times to assess calibration.

\begin{table}
\centering
\caption{The average and standard deviation of the C-index for the DPLTM, LTM and PLATM methods.}
\label{t:webtable3}
\scalebox{0.65}{
\begin{adjustbox}{center}
\setlength{\tabcolsep}{7mm}
\begin{tabular}{ccccccccc}
     \toprule
     & & & \multicolumn{3}{c}{40\% censoring rate} & \multicolumn{3}{c}{60\% censoring rate} \\
     \cmidrule(r){4-6}\cmidrule(r){7-9}
     & $r$ & $n$ & DPLTM & LTM & PLATM & DPLTM & LTM & PLATM \\    
     \midrule
     Case 1 & 0 & 1000 & 0.8374 & 0.8379 & 0.8298 & 0.8474 & 0.8475 & 0.8402 \\
     (Linear) & & & (0.0171) & (0.0167) & (0.0172) & (0.0208) & (0.0201) & (0.0209)  \\
     & & 2000 & 0.8358 & 0.8375 & 0.8334 & 0.8461 & 0.8484 & 0.8448 \\
     & & & (0.0121) & (0.0112) & (0.0113) & (0.0140) & (0.0134) & (0.0137) \\
     & 0.5 & 1000 & 0.8153 & 0.8162 & 0.8064 & 0.8281 & 0.8292 & 0.8196 \\
     & & & (0.0195) & (0.0184) & (0.0189) & (0.0229) & (0.0217) & (0.0225) \\
     & & 2000 & 0.8155 & 0.8148 & 0.8098 & 0.8221 & 0.8299 & 0.8246 \\
     & & & (0.0139) & (0.0123) & (0.0126) & (0.0152) & (0.0143) & (0.0146) \\
     & 1 & 1000 & 0.8067 & 0.8042 & 0.8106 & 0.8058 & 0.8110 & 0.8198 \\
     & & & (0.0192) & (0.0199) & (0.0200) & (0.0228) & (0.0233) & (0.0239) \\
     & & 2000 & 0.8161 & 0.8020 & 0.8062 & 0.8063 & 0.8105 & 0.8154 \\
     & & & (0.0140) & (0.0129) & (0.0130) & (0.0153) & (0.0151) & (0.0154) \\
     \\
     Case 2 & 0 & 1000 & 0.8161 & 0.7265 & 0.8251 & 0.8203 & 0.7462 & 0.8307 \\
     (Additive) & & & (0.0183) & (0.0207) & (0.0167) & (0.0224) & (0.0248) & (0.0190) \\
     & & 2000 & 0.8192 & 0.7269 & 0.8261 & 0.8255 & 0.7467 & 0.8329 \\
     & & & (0.0123) & (0.0163) & (0.0126) & (0.0146) & (0.0194) & (0.0151) \\
     & 0.5 & 1000 & 0.7896 & 0.7192 & 0.8016 & 0.7988 & 0.7360 & 0.8114 \\
     & & & (0.0218) & (0.0221) & (0.0176) & (0.0249) & (0.0262) & (0.0203) \\
     & & 2000 & 0.7945 & 0.7188 & 0.8030 & 0.8055 & 0.7358 & 0.8141 \\
     & & & (0.0137) & (0.0170) & (0.0137) & (0.0152) & (0.0202) & (0.0162) \\
     & 1 & 1000 & 0.7667 & 0.6981 & 0.7803 & 0.7792 & 0.7183 & 0.7931 \\
     & & & (0.0214) & (0.0214) & (0.0186) & (0.0250) & (0.0253) & (0.0213) \\
     & & 2000 & 0.7728 & 0.6975 & 0.7820 & 0.7860 & 0.7184 & 0.7961 \\
     & & & (0.0139) & (0.0160) & (0.0146) & (0.0162) & (0.0197) & (0.0170) \\
     \\
     Case 3 & 0 & 1000 & 0.8020 & 0.6600 & 0.7452 & 0.8023 & 0.6729 & 0.7543 \\
     (Deep) & & & (0.0235) & (0.0246) & (0.0244) & (0.0304) & (0.0284) & (0.0271) \\
     & & 2000 & 0.8096 & 0.6602 & 0.7460 & 0.8122 & 0.6737 & 0.7569 \\
     & & & (0.0147) & (0.0168) & (0.0165) & (0.0170) & (0.0198) & (0.0183) \\
     & 0.5 & 1000 & 0.7793 & 0.6516 & 0.7295 & 0.7785 & 0.6636 & 0.7398 \\
     & & & (0.0237) & (0.0258) & (0.0246) & (0.0280) & (0.0294) & (0.0282) \\
     & & 2000 & 0.7878 & 0.6528 & 0.7316 & 0.7928 & 0.6647 & 0.7434 \\
     & & & (0.0171) & (0.0180) & (0.0169) & (0.0201) & (0.0205) & (0.0192) \\
     & 1 & 1000 & 0.7547 & 0.6430 & 0.7136 & 0.7586 & 0.6540 & 0.7257 \\
     & & & (0.0236) & (0.0235) & (0.0252) & (0.0294) & (0.0293) & (0.0285) \\
     & & 2000 & 0.7657 & 0.6448 & 0.7165 & 0.7741 & 0.6553 & 0.7295 \\
     & & & (0.0166) & (0.0169) & (0.0171) & (0.0197) & (0.0201) & (0.0193) \\
     \bottomrule
\end{tabular}
\end{adjustbox}
}
\end{table}

Table~\ref{t:webtable3} exhibits the average and standard deviation of the C-index on the test data based on 200 simulation runs. Unsurprisingly, predictions obtained by the DPLTM method are comparable to or only a little worse than those by LTM and PLATM in simple settings, but DPLTM shows great superiority over the other two models under the more complex Case 3 as it produces much more accurate estimates for $\vb$ and $g$.

Tables~\ref{t:webtable4},~\ref{t:webtable5} and~\ref{t:webtable6} display the average and standard deviation of the ICI at $t_{25}$, $t_{50}$ and $t_{75}$ on the test data over 200 simulation runs. Similarly, DPLTM markedly outperforms LTM and PLATM when the true nonparametric function is highly nonlinear, and still maintains robust competitiveness compared to correctly specified models under simpler cases. Furthermore, the metric as well as its variability generally tends to increase as the time at which the calibration of models is assessed increases.

\begin{table}
\centering
\caption{The average and standard deviation of the ICI at $t_{25}$ for the DPLTM, LTM and PLATM methods.}
\label{t:webtable4}
\scalebox{0.65}{
\begin{adjustbox}{center}
\setlength{\tabcolsep}{7mm}
\begin{tabular}{ccccccccc}
     \toprule
     & & & \multicolumn{3}{c}{40\% censoring rate} & \multicolumn{3}{c}{60\% censoring rate} \\
     \cmidrule(r){4-6}\cmidrule(r){7-9}
     & $r$ & $n$ & DPLTM & LTM & PLATM & DPLTM & LTM & PLATM \\    
     \midrule
     Case 1 & 0 & 1000 & 0.0193 & 0.0178 & 0.0188 & 0.0204 & 0.0176 & 0.0191 \\
     (Linear) & & & (0.0124) & (0.0110) & (0.0111) & (0.0109) & (0.0102) & (0.0110)  \\
     & & 2000 & 0.0127 & 0.0123 & 0.0124 & 0.0129 & 0.0121 & 0.0121 \\
     & & & (0.0084) & (0.0078) & (0.0077) & (0.0082) & (0.0078) & (0.0079) \\
     & 0.5 & 1000 & 0.0314 & 0.0315 & 0.0303 & 0.0254 & 0.0238 & 0.0260 \\
     & & & (0.0142) & (0.0158) & (0.0154) & (0.0128) & (0.0120) & (0.0121) \\
     & & 2000 & 0.0262 & 0.0246 & 0.0264 & 0.0208 & 0.0191 & 0.0198 \\
     & & & (0.0093) & (0.0109) & (0.0097) & (0.0096) & (0.0096) & (0.0087) \\
     & 1 & 1000 & 0.0358 & 0.0407 & 0.0362 & 0.0320 & 0.0306 & 0.0321 \\
     & & & (0.0196) & (0.0242) & (0.0189) & (0.0138) & (0.0134) & (0.0140) \\
     & & 2000 & 0.0239 & 0.0231 & 0.0303 & 0.0231 & 0.0214 & 0.0217 \\
     & & & (0.0133) & (0.0150) & (0.0136) & (0.0101) & (0.0106) & (0.0103) \\
     \\
     Case 2 & 0 & 1000 & 0.0199 & 0.0397 & 0.0189 & 0.0208 & 0.0388 & 0.0180 \\
     (Additive) & & & (0.0133) & (0.0187) & (0.0110) & (0.0109) & (0.0123) & (0.0108) \\
     & & 2000 & 0.0127 & 0.0366 & 0.0113 & 0.0127 & 0.0248 & 0.0112 \\
     & & & (0.0085) & (0.0123) & (0.0077) & (0.0078) & (0.0125) & (0.0069) \\
     & 0.5 & 1000 & 0.0343 & 0.0471 & 0.0288 & 0.0284 & 0.0351 & 0.0240 \\
     & & & (0.0192) & (0.0217) & (0.0129) & (0.0151) & (0.0183) & (0.0128) \\
     & & 2000 & 0.0237 & 0.0290 & 0.0220 & 0.0199 & 0.0253 & 0.0186 \\
     & & & (0.0119) & (0.0127) & (0.0095) & (0.0100) & (0.0131) & (0.0091) \\
     & 1 & 1000 & 0.0349 & 0.0420 & 0.0341 & 0.0339 & 0.0422 & 0.0310 \\
     & & & (0.0172) & (0.0189) & (0.0144) & (0.0150) & (0.0233) & (0.0135) \\
     & & 2000 & 0.0228 & 0.0290 & 0.0221 & 0.0223 & 0.0301 & 0.0222 \\
     & & & (0.0117) & (0.0145) & (0.0094) & (0.0103) & (0.0166) & (0.0101) \\
     \\
     Case 3 & 0 & 1000 & 0.0210 & 0.0430 & 0.0409 & 0.0206 & 0.0415 & 0.0362 \\
     (Deep) & & & (0.0136) & (0.0236) & (0.0229) & (0.0127) & (0.0218) & (0.0190) \\
     & & 2000 & 0.0139 & 0.0409 & 0.0369 & 0.0143 & 0.0342 & 0.0307 \\
     & & & (0.0091) & (0.0192) & (0.0182) & (0.0084) & (0.0181) & (0.0149) \\
     & 0.5 & 1000 & 0.0334 & 0.0407 & 0.0394 & 0.0266 & 0.0354 & 0.0403 \\
     & & & (0.0152) & (0.0212) & (0.0187) & (0.0149) & (0.0184) & (0.0215) \\
     & & 2000 & 0.0267 & 0.0335 & 0.0296 & 0.0229 & 0.0321 & 0.0318 \\
     & & & (0.0135) & (0.0162) & (0.0147) & (0.0112) & (0.0131) & (0.0147) \\
     & 1 & 1000 & 0.0326 & 0.0411 & 0.0425 & 0.0336 & 0.0373 & 0.0410 \\
     & & & (0.0165) & (0.0200) & (0.0216) & (0.0160) & (0.0248) & (0.0247) \\
     & & 2000 & 0.0215 & 0.0316 & 0.0302 & 0.0251 & 0.0299 & 0.0328 \\
     & & & (0.0123) & (0.0159) & (0.0157) & (0.0124) & (0.0208) & (0.0176) \\
     \bottomrule
\end{tabular}
\end{adjustbox}
}
\end{table}

\begin{table}
\centering
\caption{The average and standard deviation of the ICI at $t_{50}$ for the DPLTM, LTM and PLATM methods.}
\label{t:webtable5}
\scalebox{0.65}{
\begin{adjustbox}{center}
\setlength{\tabcolsep}{7mm}
\begin{tabular}{ccccccccc}
     \toprule
     & & & \multicolumn{3}{c}{40\% censoring rate} & \multicolumn{3}{c}{60\% censoring rate} \\
     \cmidrule(r){4-6}\cmidrule(r){7-9}
     & $r$ & $n$ & DPLTM & LTM & PLATM & DPLTM & LTM & PLATM \\    
     \midrule
     Case 1 & 0 & 1000 & 0.0249 & 0.0220 & 0.0257 & 0.0238 & 0.0244 & 0.0257 \\
     (Linear) & & & (0.0167) & (0.0131) & (0.0134) & (0.0147) & (0.0149) & (0.0148)  \\
     & & 2000 & 0.0163 & 0.0154 & 0.0156 & 0.0169 & 0.0160 & 0.0158 \\
     & & & (0.0105) & (0.0096) & (0.0098) & (0.0109) & (0.0098) & (0.0100) \\
     & 0.5 & 1000 & 0.0349 & 0.0334 & 0.0385 & 0.0315 & 0.0310 & 0.0324 \\
     & & & (0.0201) & (0.0199) & (0.0204) & (0.0168) & (0.0161) & (0.0167) \\
     & & 2000 & 0.0286 & 0.0238 & 0.0275 & 0.0224 & 0.0214 & 0.0209 \\
     & & & (0.0146) & (0.0108) & (0.0155) & (0.0118) & (0.0111) & (0.0112) \\
     & 1 & 1000 & 0.0408 & 0.0399 & 0.0419 & 0.0356 & 0.0338 & 0.0360 \\
     & & & (0.0187) & (0.0242) & (0.0181) & (0.0169) & (0.0179) & (0.0184) \\
     & & 2000 & 0.0250 & 0.0303 & 0.0269 & 0.0240 & 0.0233 & 0.0248 \\
     & & & (0.0136) & (0.0199) & (0.0132) & (0.0102) & (0.0121) & (0.0112) \\
     \\
     Case 2 & 0 & 1000 & 0.0274 & 0.0457 & 0.0241 & 0.0275 & 0.0436 & 0.0244 \\
     (Additive) & & & (0.0149) & (0.0237) & (0.0140) & (0.0129) & (0.0166) & (0.0150) \\
     & & 2000 & 0.0172 & 0.0343 & 0.0145 & 0.0173 & 0.0302 & 0.0151 \\
     & & & (0.0103) & (0.0163) & (0.0093) & (0.0106) & (0.0162) & (0.0104) \\
     & 0.5 & 1000 & 0.0402 & 0.0515 & 0.0392 & 0.0354 & 0.0477 & 0.0302 \\
     & & & (0.0234) & (0.0247) & (0.0246) & (0.0177) & (0.0245) & (0.0167) \\
     & & 2000 & 0.0283 & 0.0358 & 0.0297 & 0.0229 & 0.0309 & 0.0208 \\
     & & & (0.0169) & (0.0136) & (0.0166) & (0.0117) & (0.0178) & (0.0112) \\
     & 1 & 1000 & 0.0425 & 0.0489 & 0.0400 & 0.0344 & 0.0502 & 0.0343 \\
     & & & (0.0182) & (0.0235) & (0.0209) & (0.0197) & (0.0257) & (0.0164) \\
     & & 2000 & 0.0266 & 0.0411 & 0.0310 & 0.0292 & 0.0361 & 0.0223 \\
     & & & (0.0106) & (0.0227) & (0.0156) & (0.0141) & (0.0182) & (0.0121) \\
     \\
     Case 3 & 0 & 1000 & 0.0274 & 0.0549 & 0.0503 & 0.0276 & 0.0553 & 0.0501 \\
     (Deep) & & & (0.0185) & (0.0252) & (0.0240) & (0.0163) & (0.0265) & (0.0265) \\
     & & 2000 & 0.0193 & 0.0481 & 0.0357 & 0.0182 & 0.0462 & 0.0333 \\
     & & & (0.0128) & (0.0185) & (0.0175) & (0.0116) & (0.0221) & (0.0186) \\
     & 0.5 & 1000 & 0.0425 & 0.0484 & 0.0510 & 0.0342 & 0.0543 & 0.0474 \\
     & & & (0.0190) & (0.0264) & (0.0272) & (0.0184) & (0.0224) & (0.0230) \\
     & & 2000 & 0.0292 & 0.0375 & 0.0345 & 0.0247 & 0.0404 & 0.0306 \\
     & & & (0.0125) & (0.0200) & (0.0219) & (0.0133) & (0.0168) & (0.0180) \\
     & 1 & 1000 & 0.0424 & 0.0528 & 0.0491 & 0.0399 & 0.0518 & 0.0500 \\
     & & & (0.0231) & (0.0271) & (0.0225) & (0.0213) & (0.0264) & (0.0273) \\
     & & 2000 & 0.0293 & 0.0361 & 0.0351 & 0.0295 & 0.0432 & 0.0339 \\
     & & & (0.0130) & (0.0182) & (0.0165) & (0.0154) & (0.0219) & (0.0181) \\
     \bottomrule
\end{tabular}
\end{adjustbox}
}
\end{table}

\begin{table}
\centering
\caption{The average and standard deviation of the ICI at $t_{75}$ for the DPLTM, LTM and PLATM methods.}
\label{t:webtable6}
\scalebox{0.65}{
\begin{adjustbox}{center}
\setlength{\tabcolsep}{7mm}
\begin{tabular}{ccccccccc}
     \toprule
     & & & \multicolumn{3}{c}{40\% censoring rate} & \multicolumn{3}{c}{60\% censoring rate} \\
     \cmidrule(r){4-6}\cmidrule(r){7-9}
     & $r$ & $n$ & DPLTM & LTM & PLATM & DPLTM & LTM & PLATM \\    
     \midrule
     Case 1 & 0 & 1000 & 0.0289 & 0.0258 & 0.0293 & 0.0296 & 0.0290 & 0.0314 \\
     (Linear) & & & (0.0169) & (0.0156) & (0.0163) & (0.0172) & (0.0178) & (0.0188)  \\
     & & 2000 & 0.0192 & 0.0186 & 0.0188 & 0.0213 & 0.0197 & 0.0193 \\
     & & & (0.0113) & (0.0114) & (0.0118) & (0.0135) & (0.0125) & (0.0126) \\
     & 0.5 & 1000 & 0.0364 & 0.0324 & 0.0403 & 0.0343 & 0.0381 & 0.0369 \\
     & & & (0.0226) & (0.0169) & (0.0221) & (0.0189) & (0.0197) & (0.0194) \\
     & & 2000 & 0.0248 & 0.0293 & 0.0288 & 0.0272 & 0.0261 & 0.0259 \\
     & & & (0.0114) & (0.0097) & (0.0170) & (0.0122) & (0.0136) & (0.0133) \\
     & 1 & 1000 & 0.0420 & 0.0494 & 0.0488 & 0.0405 & 0.0426 & 0.0415 \\
     & & & (0.0215) & (0.0264) & (0.0248) & (0.0207) & (0.0224) & (0.0214) \\
     & & 2000 & 0.0267 & 0.0276 & 0.0307 & 0.0257 & 0.0263 & 0.0290 \\
     & & & (0.0149) & (0.0167) & (0.0152) & (0.0136) & (0.0147) & (0.0143) \\
     \\
     Case 2 & 0 & 1000 & 0.0270 & 0.0472 & 0.0287 & 0.0336 & 0.0466 & 0.0277 \\
     (Additive) & & & (0.0104) & (0.0287) & (0.0160) & (0.0141) & (0.0267) & (0.0184) \\
     & & 2000 & 0.0216 & 0.0471 & 0.0187 & 0.0244 & 0.0357 & 0.0188 \\
     & & & (0.0082) & (0.0208) & (0.0100) & (0.0117) & (0.0173) & (0.0116) \\
     & 0.5 & 1000 & 0.0291 & 0.0530 & 0.0424 & 0.0293 & 0.0506 & 0.0361 \\
     & & & (0.0142) & (0.0259) & (0.0229) & (0.0136) & (0.0301) & (0.0206) \\
     & & 2000 & 0.0230 & 0.0395 & 0.0325 & 0.0268 & 0.0389 & 0.0266 \\
     & & & (0.0073) & (0.0163) & (0.0171) & (0.0096) & (0.0232) & (0.0140) \\
     & 1 & 1000 & 0.0414 & 0.0510 & 0.0456 & 0.0401 & 0.0589 & 0.0397 \\
     & & & (0.0279) & (0.0336) & (0.0267) & (0.0228) & (0.0299) & (0.0198) \\
     & & 2000 & 0.0245 & 0.0362 & 0.0359 & 0.0287 & 0.0410 & 0.0299 \\
     & & & (0.0158) & (0.0217) & (0.0182) & (0.0139) & (0.0234) & (0.0156) \\
     \\
     Case 3 & 0 & 1000 & 0.0312 & 0.0550 & 0.0505 & 0.0332 & 0.0587 & 0.0534 \\
     (Deep) & & & (0.0189) & (0.0275) & (0.0259) & (0.0191) & (0.0320) & (0.0277) \\
     & & 2000 & 0.0226 & 0.0517 & 0.0391 & 0.0248 & 0.0550 & 0.0364 \\
     & & & (0.0128) & (0.0236) & (0.0192) & (0.0147) & (0.0223) & (0.0195) \\
     & 0.5 & 1000 & 0.0451 & 0.0488 & 0.0485 & 0.0440 & 0.0601 & 0.0530 \\
     & & & (0.0216) & (0.0294) & (0.0288) & (0.0203) & (0.0256) & (0.0243) \\
     & & 2000 & 0.0326 & 0.0403 & 0.0433 & 0.0291 & 0.0446 & 0.0365 \\
     & & & (0.0155) & (0.0246) & (0.0240) & (0.0138) & (0.0177) & (0.0184) \\
     & 1 & 1000 & 0.0423 & 0.0530 & 0.0517 & 0.0451 & 0.0585 & 0.0565 \\
     & & & (0.0240) & (0.0263) & (0.0264) & (0.0228) & (0.0326) & (0.0284) \\
     & & 2000 & 0.0271 & 0.0346 & 0.0360 & 0.0303 & 0.0446 & 0.0334 \\
     & & & (0.0161) & (0.0196) & (0.0212) & (0.0169) & (0.0245) & (0.0189) \\
     \bottomrule
\end{tabular}
\end{adjustbox}
}
\end{table}

\subsection{Comparison between DPLTM and DPLCM}
\label{s:comparison}

We make a comprehensive comparison between our DPLTM method and the DPLCM method proposed by \citet{Zhong2022} in both estimation and prediction. The partially linear Cox model can be represented by its conditional hazard function with the form of
\begin{align}
\label{eq:plcm}
\lambda(u\vert\vz,\vx)=\lambda_0(u)\exp\left\{\vb^{\top}\vz+g(\vx)\right\},
\end{align}
where $\lambda_0$ is an unknown baseline hazard function. Given $\{\vv_i=(T_i,\Delta_i,\vz_i,\vx_i),\ i=1,\cdots,n\}$, the parameter vector $\vb$ and the nonparametric function $g$ can be estimated by maximizing the log partial likelihood \citep{cox1975partial}
\begin{align*}
(\widehat{\vb},\widehat{g})=\underset{(\vb,g)\in\R^p\times\mathcal{G}}{\arg\max}\mathscr{L}_n(\vb,g),
\end{align*}
where $\mathscr{L}_n(\vb,g)=\sum_{i=1}^{n}\Delta_{i}\left[\vb^{\top}\vz_i+g(\vx_i)-\log\sum_{j:T_{j}\geq T_{i}}\exp\left\{\vb^{\top}\vz_j+g(\vx_j)\right\}\right]$. Moreover, the estimate of the cumulative baseline hazard function $\Lambda_0(t)=\int_0^t\lambda_0(s)ds$ is further given by the Breslow estimator \citep{breslow1972discussion} as
\begin{align*}
\widehat{\Lambda}_0(t)=\sum_{i=1}^n\frac{\Delta_iI(T_i\le t)}{\sum_{j:T_j\ge T_i}\exp\left\{\widehat{\vb}^{\top}\vz_j+\widehat{g}(\vx_j)\right\}}.
\end{align*}
Then the predicted probability of the outcome prior to $t_0$ can be calculated as $\widehat{P}_i^{t_0}=1-\exp\left\{-\widehat{\Lambda}_0(t_0)\exp\left\{\widehat{\vb}^{\top}\vz_i+\widehat{g}(\vx_i)\right\}\right\}$. On the other hand, the Cox proportional hazards model can be seen as a particular case of the class of semiparametric transformation models. In fact, (\ref{eq:plcm}) can be restated as
\begin{align*}
\log\Lambda_0(U)=-\vb^{\top}\vz-g(\vx)+\epsilon,
\end{align*}
where the error term $\epsilon$ follows the extreme value distribution. It is easy to see that the term $\log \Lambda_0(U)$ in the Cox model serves the role of $H(U)$ in the class of transformation models. Therefore, we can compute all the evaluation metrics that have been mentioned previously for the DPLTM and DPLCM methods, and then assess their estimation accuracy and predictive power across various configurations. We only carry out simulations for Case 3 of $g_0$ since we are comparing two DNN-based models.

Table~\ref{t:webtable7} presents a summary of the estimation accuracy of DPLTM and DPLCM. It is not surprising that DPLCM does slightly better than DPLTM with regard to all evaluation metrics when $r=0$, i.e. the true model is exactly the Cox proportional hazards model. But DPLTM substantially outperforms DPLCM in the case of $r=0.5$ or 1, and the performance gap becomes broader when $r$ increases from 0.5 to 1.

\begin{table}
\centering
\caption{Comparison of estimation accuracy between DPLTM and DPLCM.}
\label{t:webtable7}
\scalebox{0.6}{
\begin{adjustbox}{center}
\setlength{\tabcolsep}{7mm}
\begin{tabular}{ccccccccc}
     \toprule
     & & & \multicolumn{2}{c}{$r=0$} & \multicolumn{2}{c}{$r=0.5$} & \multicolumn{2}{c}{$r=1$}\\
     \cmidrule(r){4-5}\cmidrule(r){6-7}\cmidrule(r){8-9}
     & Censoring rate & $n$ & DPLTM & DPLCM & DPLTM & DPLCM & DPLTM & DPLCM \\    
     \midrule
     The bias and standard & 40\% & 1000 & -0.0395 & -0.0306 & -0.0457 & -0.1975 & -0.0570 & -0.3033 \\
     deviation of $\widehat{\beta}_1$ & & & (0.1012) & (0.1057) & (0.1293) & (0.1108) & (0.1544) & (0.1109)  \\
     & & 2000 & -0.0322 & -0.0275 & -0.0350 & -0.2186 & -0.0344 & -0.3339 \\
     & & & (0.0683) & (0.0733) & (0.0896) & (0.0770) & (0.1012) & (0.0779) \\
     & 60\% & 1000 & -0.0474 & -0.0460 & -0.0586 & -0.1449 & -0.0463 & -0.2399 \\
     & & & (0.1239) & (0.1393) & (0.1577) & (0.1430) & (0.1764) & (0.1402) \\
     & & 2000 & -0.0286 & -0.0314 & -0.0478 & -0.1708 & -0.0378 & -0.2698 \\
     & & & (0.0833) & (0.0920) & (0.1022) & (0.0940) & (0.1138) & (0.0948) \\
     \\
     The bias and standard & 40\% & 1000 & 0.0466 & 0.0340 & 0.0409 & 0.1952 & 0.0375 & 0.3037 \\
     deviation of $\widehat{\beta}_2$ & & & (0.0982) & (0.1067) & (0.1242) & (0.11057) & (0.1450) & (0.1075) \\
     & & 2000 & 0.0389 & 0.0267 & 0.0265 & 0.2206 & 0.0245 & 0.3360 \\
     & & & (0.0720) & (0.0749) & (0.0924) & (0.0743) & (0.1028) & (0.0761) \\
     & 60\% & 1000 & 0.0559 & 0.0374 & 0.0382 & 0.1431 & 0.0438 & 0.2418 \\
     & & & (0.1186) & (0.1291) & (0.1473) & (0.1309) & (0.1680) & (0.1344) \\
     & & 2000 & 0.0406 & 0.0280 & 0.0244 & 0.1612 & 0.0299 & 0.2645 \\
     & & & (0.0828) & (0.0888) & (0.1007) & (0.0907) & (0.1140) & (0.0918) \\
     \\
     The empirical coverage & 40\% & 1000 & 0.925 & 0.945 & 0.925 & 0.470 & 0.930 & 0.160 \\
     probability of 95\% & & 2000 & 0.945 & 0.940 & 0.920 & 0.145 & 0.925 & 0.010 \\
     confidence intervals for $\beta_{01}$ & 60\% & 1000 & 0.955 & 0.925 & 0.915 & 0.745 & 0.915 & 0.470 \\
     & & 2000 & 0.920 & 0.950 & 0.920 & 0.450 & 0.925 & 0.145 \\
     \\
     The empirical coverage & 40\% & 1000 & 0.935 & 0.920 & 0.935 & 0.465 & 0.955 & 0.150 \\
     probability of 95\% & & 2000 & 0.920 & 0.940 & 0.925 & 0.125 & 0.940 & 0.010 \\
     confidence intervals for $\beta_{02}$ & 60\% & 1000 & 0.915 & 0.955 & 0.935 & 0.770 & 0.950 & 0.455 \\
     & & 2000 & 0.935 & 0.950 & 0.915 & 0.485 & 0.955 & 0.125 \\
     \\
     The average and & 40\% & 1000 & 0.4069 & 0.3382 & 0.4032 & 0.5705 & 0.4516 & 0.7333 \\
     standard deviation of& & & (0.0549) & (0.0434) & (0.0696) & (0.0563) & (0.0624) & (0.0842) \\
     the relative error of $\widehat{g}$ & & 2000 & 0.3421 & 0.2796 & 0.3590 & 0.5130 & 0.3788 & 0.7080 \\
     & & & (0.0416) & (0.0305) & (0.0437) & (0.0439) & (0.0487) & (0.0510) \\
     & 60\% & 1000 & 0.4287 & 0.4027 & 0.4739 & 0.5944 & 0.4835 & 0.7678 \\
     & & & (0.0759) & (0.0633) & (0.0890) & (0.0712) & (0.0851) & (0.0954) \\
     & & 2000 & 0.3672 & 0.3043 & 0.4186 & 0.5478 & 0.4390 & 0.7485 \\
     & & & (0.0593) & (0.0457) & (0.0567) & (0.0482) & (0.0559) & (0.0664) \\
     \\
     The average and & 40\% & 1000 & 0.0508 & 0.0416 & 0.0501 & 0.1881 & 0.0558 & 0.2187 \\
     standard deviation of the & & & (0.0328) & (0.0287) & (0.0378) & (0.0516) & (0.0392) & (0.0628)\\
     WISE of $\widehat{H}(t)$ or $\log \widehat{\Lambda}_0(t)$ & & 2000 & 0.0356 & 0.0265 & 0.0382 & 0.1584 & 0.0375 & 0.2065 \\
     & & & (0.0190) & (0.0183) & (0.0245) & (0.0297) & (0.0267) & (0.0401)\\
     & 60\% & 1000 & 0.0542 & 0.0511 & 0.0576 & 0.1407 & 0.0578 & 0.1918 \\
     & & & (0.0335) & (0.0376) & (0.0447) & (0.0492) & (0.0434) & (0.0763)\\
     & & 2000 & 0.0362 & 0.0312 & 0.0364 & 0.1351 & 0.0459 & 0.1942 \\
     & & & (0.0216) & (0.0248) & (0.0301) & (0.0271) & (0.0291) & (0.0508)\\
     \bottomrule
\end{tabular}
\end{adjustbox}
}
\end{table}

\begin{table}
\centering
\caption{Comparison of predictive power between DPLTM and DPLCM.}
\label{t:webtable8}
\scalebox{0.6}{
\begin{adjustbox}{center}
\setlength{\tabcolsep}{7mm}
\begin{tabular}{ccccccccc}
     \toprule
     & & & \multicolumn{2}{c}{$r=0$} & \multicolumn{2}{c}{$r=0.5$} & \multicolumn{2}{c}{$r=1$}\\
     \cmidrule(r){4-5}\cmidrule(r){6-7}\cmidrule(r){8-9}
     & Censoring rate & $n$ & DPLTM & DPLCM & DPLTM & DPLCM & DPLTM & DPLCM \\    
     \midrule
     The average and & 40\% & 1000 & 0.8020 & 0.8045 & 0.7793 & 0.7786 & 0.7547 & 0.7542 \\
     standard deviation & & & (0.0235) & (0.0208) & (0.0237) & (0.0222) & (0.0236) & (0.0244)  \\
     of the C-index & & 2000 & 0.8096 & 0.8104 & 0.7878 & 0.7870 & 0.7657 & 0.7672 \\
     & & & (0.0147) & (0.0126) & (0.0171) & (0.0141) & (0.0166) & (0.0158) \\
     & 60\% & 1000 & 0.8023 & 0.8035 & 0.7785 & 0.7811 & 0.7586 & 0.7623 \\
     & & & (0.0304) & (0.0234) & (0.0280) & (0.0262) & (0.0294) & (0.0283) \\
     & & 2000 & 0.8122 & 0.8137 & 0.7928 & 0.7942 & 0.7741 & 0.7735 \\
     & & & (0.0170) & (0.0162) & (0.0201) & (0.0170) & (0.0197) & (0.0173) \\
     \\
     The average and & 40\% & 1000 & 0.0210 & 0.0193 & 0.0326 & 0.0411 & 0.0334 & 0.0440 \\
     standard deviation & & & (0.0136) & (0.0107) & (0.0152) & (0.0203) & (0.0165) & (0.0235) \\
     of the ICI at $t_{25}$ & & 2000 & 0.0139 & 0.0130 & 0.0267 & 0.0320 & 0.0215 & 0.0282 \\
     & & & (0.0091) & (0.0070) & (0.0135) & (0.0168) & (0.0123) & (0.0137) \\
     & 60\% & 1000 & 0.0206 & 0.0168 & 0.0266 & 0.0354 & 0.0336 & 0.0428 \\
     & & & (0.0127) & (0.0102) & (0.0149) & (0.0161) & (0.0160) & (0.0194) \\
     & & 2000 & 0.0143 & 0.0147 & 0.0229 & 0.0281 & 0.0251 & 0.0357 \\
     & & & (0.0084) & (0.0071) & (0.0112) & (0.0127) & (0.0124) & (0.0175) \\
     \\
     The average and & 40\% & 1000 & 0.0274 & 0.0241 & 0.0425 & 0.0489 & 0.0424 & 0.0503 \\
     standard deviation & & & (0.0185) & (0.0113) & (0.0190) & (0.0292) & (0.0231) & (0.0256) \\
     of the ICI at $t_{50}$ & & 2000 & 0.0193 & 0.0161 & 0.0292 & 0.0342 & 0.0293 & 0.0366 \\
     & & & (0.0108) & (0.0083) & (0.0125) & (0.0162) & (0.0130) & (0.0205) \\
     & 60\% & 1000 & 0.0276 & 0.0219 & 0.0342 & 0.0418 & 0.0399 & 0.0515 \\
     & & & (0.0163) & (0.0117) & (0.0184) & (0.0227) & (0.0213) & (0.0279) \\
     & & 2000 & 0.0182 & 0.0168 & 0.0247 & 0.0345 & 0.0295 & 0.0402 \\
     & & & (0.0116) & (0.0087) & (0.0133) & (0.0174) & (0.0154) & (0.0228) \\
     \\
     The average and & 40\% & 1000 & 0.0312 & 0.0265 & 0.0451 & 0.0507 & 0.0423 & 0.0521 \\
     standard deviation & & & (0.0189) & (0.0157) & (0.0216) & (0.0296) & (0.0240) & (0.0283) \\
     of the ICI at $t_{75}$ & & 2000 & 0.0226 & 0.0196 & 0.0326 & 0.0384 & 0.0271 & 0.0356 \\
     & & & (0.0128) & (0.0119) & (0.0155) & (0.0218) & (0.0161) & (0.0192) \\
     & 60\% & 1000 & 0.0332 & 0.0253 & 0.0440 & 0.0485 & 0.0451 & 0.0530 \\
     & & & (0.0191) & (0.0140) & (0.0203) & (0.0264) & (0.0228) & (0.0308) \\
     & & 2000 & 0.0248 & 0.0211 & 0.0291 & 0.0377 & 0.0303 & 0.0417 \\
     & & & (0.0147) & (0.0114) & (0.0138) & (0.0196) & (0.0169) & (0.0243) \\
     \bottomrule
\end{tabular}
\end{adjustbox}
}
\end{table} 

Table~\ref{t:webtable8} exhibits the prediction power of the two methods. The C-index values for DPLCM are comparable to those for DPLTM in all simulation settings. However, in terms of the calibration metric ICI, DPLCM is incapable of competing with DPLTM when the proportional hazards assumption is not satisfied for the underlying model, which implies that DPLTM generally enables more reliable predictions.

\subsection{Prediction results for the SEER lung cancer dataset}
\label{s:predapp}
We further validate the predictive ability of the DPLTM method by comparing it with other methods, including traditional methods LTM and PLATM, machine learning methods random survival forest (RSF) and survival support vector machine (SSVM), and the DNN-based method DPLCM on the SEER lung cancer dataset using the C-index and the ICI as evaluation metrics. Our method results in a C-index value of 0.7028, outperforming all other methods (LTM: 0.6582, PLATM: 0.6775, RSF: 0.6927, SSVM: 0.6699, DPLCM: 0.6974). 

For the time-dependent calibration metric ICI, it is computed at the $k$-th month post admission, $1\le k\le 80$, since the maximum of all observed event times is 83 months, and roughly 95\% of the times are no more than 80 months. The SSVM method is omitted from the comparison in terms of ICI, as it can only predict a risk score instead of a survival function for each individual, making it difficult to assess calibration. Web Figure~\ref{f:ICI} plots the ICI values across 80 months for all methods except SSVM. The results indicate that DPLTM provides the most accurate predictions for this dataset most of the time.

\begin{figure}[!h]
\centering
\includegraphics[width=0.8\textwidth]{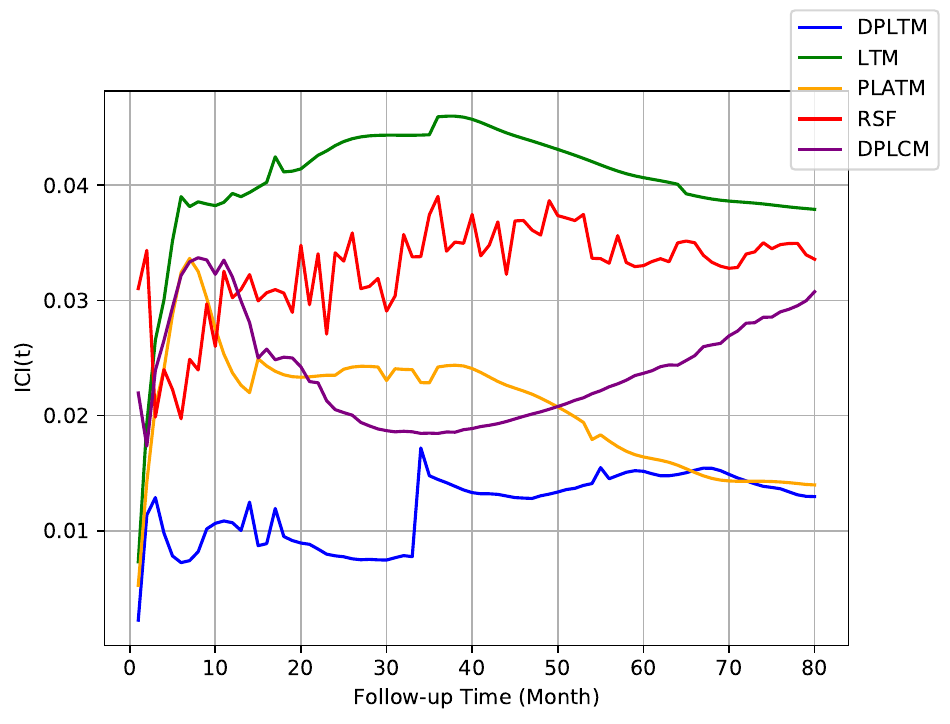}
\caption{The ICI values across 80 months on the SEER lung cancer dataset for all methods except SSVM.}
\label{f:ICI}
\end{figure}

\section{Further simulation studies}
\label{s:addexp}

\subsection{Hypothesis testing}

As in the real data application, we carry out a hypothesis test in simulation studies to investigate whether the linearly modelled covariates are significantly associated with the survival time, and how well the three methods can detect such relationships under finite sample situations. For simplicity, we only test the significance of $\beta_1$, i.e. the first component of the parameter vector. We consider the following testing problem:
\begin{align*}
H_0:\beta_1=0\quad\text{vs.}\quad H_1:\beta_1\neq0.
\end{align*}
The test statistic and the criterion for rejecting the null hypothesis $H_0$ are the same as in Section 5 of the main article.

The simulation setups are all identical to those in Section 4 of the main article, except that the true value of $\beta_1$, denoted by $\beta_{01}$, is set to be 0, 0.1, 0.3 and 1, respectively. The nominal significance level $\alpha$ is chosen as 0.05 standardly. When $\beta_{01}$ takes the value 0, we obtain the size of the test empirically as the proportion of the simulation runs where we falsely reject the null hypothesis. Otherwise, we calculate the empirical power of the test in a similar way. For convenience, we again only consider Case 3 of $g_0$.

Table \ref{t:webtable9} reports the empirically estimated size and power for the three methods. When data are generated according to $H_0$, i.e. $\beta_{01}$=0, the DPLTM method yields empirical sizes that are generally close to 0.05, and performs moderately better than LTM and PLATM. When $\beta_{01}$=0.1 or 0.3, the estimated power values for the DPLTM method are substantially higher than those for the other two methods, suggesting the effectiveness of our method in identifying the relationship. When $\beta_{01}$=1, all three methods lead to a rejection rate of 100\% in all situations considered, which is expected because the estimation bias is markedly outweighed by the large deviation from the null hypothesis.

\begin{table}
\centering
\caption{The empirical size and power of the hypothesis test for the DPLTM, LTM and PLATM methods.}
\label{t:webtable9}
\scalebox{0.6}{
\begin{adjustbox}{center}
\setlength{\tabcolsep}{7mm}
\begin{tabular}{ccccccccc}
     \toprule
     & & & \multicolumn{3}{c}{40\% censoring rate} & \multicolumn{3}{c}{60\% censoring rate} \\
     \cmidrule(r){4-6}\cmidrule(r){7-9}
     $\beta_{01}$ & $r$ & $n$ & DPLTM & LTM & PLATM & DPLTM & LTM & PLATM \\    
     \midrule
     0 & 0 & 1000 & 0.030 & 0.045 & 0.045 & 0.040 & 0.060 & 0.055 \\
     & & 2000 & 0.035 & 0.060 & 0.085 & 0.055 & 0.070 & 0.090 \\
     & 0.5 & 1000 & 0.045 & 0.050 & 0.055 & 0.035 & 0.040 & 0.060 \\
     & & 2000 & 0.045 & 0.070 & 0.080 & 0.050 & 0.075 & 0.075 \\
     & 1 & 1000 & 0.055 & 0.045 & 0.070 & 0.045 & 0.050 & 0.055 \\
     & & 2000 & 0.045 & 0.080 & 0.085 & 0.060 & 0.065 & 0.075 \\
     \\
     0.1 & 0 & 1000 & 0.190 & 0.115 & 0.115 & 0.140 & 0.115 & 0.125 \\
     & & 2000 & 0.305 & 0.160 & 0.140 & 0.205 & 0.160 & 0.165 \\
     & 0.5 & 1000 & 0.180 & 0.125 & 0.115 & 0.100 & 0.090 & 0.095 \\
     & & 2000 & 0.205 & 0.140 & 0.135 & 0.175 & 0.115 & 0.125 \\
     & 1 & 1000 & 0.140 & 0.120 & 0.110 & 0.130 & 0.100 & 0.125 \\
     & & 2000 & 0.150 & 0.115 & 0.120 & 0.145 & 0.115 & 0.120 \\
     \\
     0.3 & 0 & 1000 & 0.875 & 0.520 & 0.570 & 0.710 & 0.470 & 0.545 \\
     & & 2000 & 1.000 & 0.830 & 0.835 & 0.915 & 0.745 & 0.735 \\
     & 0.5 & 1000 & 0.740 & 0.520 & 0.525 & 0.550 & 0.425 & 0.450 \\
     & & 2000 & 0.970 & 0.790 & 0.800 & 0.865 & 0.695 & 0.695 \\
     & 1 & 1000 & 0.625 & 0.470 & 0.465 & 0.495 & 0.390 & 0.445 \\
     & & 2000 & 0.870 & 0.740 & 0.745 & 0.780 & 0.640 & 0.665 \\
     \\
     1 & 0 & 1000 & 1.000 & 1.000 & 1.000 & 1.000 & 1.000 & 1.000 \\
     & & 2000 & 1.000 & 1.000 & 1.000 & 1.000 & 1.000 & 1.000 \\
     & 0.5 & 1000 & 1.000 & 1.000 & 1.000 & 1.000 & 1.000 & 1.000 \\
     & & 2000 & 1.000 & 1.000 & 1.000 & 1.000 & 1.000 & 1.000 \\
     & 1 & 1000 & 1.000 & 1.000 & 1.000 & 1.000 & 1.000 & 1.000 \\
     & & 2000 & 1.000 & 1.000 & 1.000 & 1.000 & 1.000 & 1.000 \\
     \bottomrule
\end{tabular}
\end{adjustbox}
}
\end{table}

\subsection{Sensitivity analysis}

\begin{table}
\centering
\caption{The bias and standard deviation of $\widehat{\beta}_1$, and the average and standard deviation of the C-index in all three scenarios considered in the sensitivity analysis.}
\label{t:webtable10}
\scalebox{0.6}{
\begin{adjustbox}{center}
\setlength{\tabcolsep}{7mm}
\begin{tabular}{ccccccccc}
     \toprule
     & & & \multicolumn{3}{c}{40\% censoring rate} & \multicolumn{3}{c}{60\% censoring rate}\\
     \cmidrule(r){4-6}\cmidrule(r){7-9}
     & $r$ & $n$ & Scenario 1 & Scenario 2 & Scenario 3 & Scenario 1 & Scenario 2 & Scenario 3 \\    
     \midrule
     The bias and standard & 0 & 1000 & -0.0395 & -0.1420 & -0.3245 & -0.0474 & -0.1548 & -0.2769 \\
     deviation of $\widehat{\beta}_1$ & & & (0.1012) & (0.1020) & (0.0954) & (0.1239) & (0.1236) & (0.1232)  \\
     & & 2000 & -0.0322 & -0.1259 & -0.3332 & -0.0286 & -0.1387 & -0.2877 \\
     & & & (0.0683) & (0.0722) & (0.0701) & (0.0833) & (0.0867) & (0.0902) \\
     & 0.5 & 1000 & -0.0457 & -0.1272 & -0.2288 & -0.0586 & -0.1427 & -0.2016 \\
     & & & (0.1293) & (0.1284) & (0.1186) & (0.1577) & (0.1582) & (0.1435) \\
     & & 2000 & -0.0350 & -0.1175 & -0.2369 & -0.0478 & -0.1297 & -0.2169 \\
     & & & (0.0896) & (0.0884) & (0.0879) & (0.1022) & (0.1046) & (0.1053) \\
     & 1 & 1000 & -0.0570 & -0.1093 & -0.1834 & -0.0463 & -0.1326 & -0.1753 \\
     & & & (0.1544) & (0.1555) & (0.1417) & (0.1764) & (0.1746) & (0.1588) \\
     & & 2000 & -0.0344 & -0.0988 & -0.1921 & -0.0378 & -0.1174 & -0.1897 \\
     & & & (0.1012) & (0.0997) & (0.1001) & (0.1138) & (0.1164) & (0.1161) \\
     \\

     The average and & 0 & 1000 & 0.8020 & 0.7825 & 0.7251 & 0.8023 & 0.7809 & 0.7358 \\
     standard deviation of& & & (0.0235) & (0.0221) & (0.0222) & (0.0304) & (0.0257) & (0.0267) \\
     the C-index & & 2000 & 0.8096 & 0.7913 & 0.7298 & 0.8122 & 0.7932 & 0.7422 \\
     & & & (0.0147) & (0.0135) & (0.0161) & (0.0170) & (0.0179) & (0.0187) \\
     & 0.5 & 1000 & 0.7793 & 0.7613 & 0.7081 & 0.7785 & 0.7593 & 0.7199 \\
     & & & (0.0237) & (0.0223) & (0.0246) & (0.0280) & (0.0284) & (0.0278) \\
     & & 2000 & 0.7878 & 0.7711 & 0.7150 & 0.7928 & 0.7758 & 0.7269 \\
     & & & (0.0171) & (0.0154) & (0.0161) & (0.0201) & (0.0179) & (0.0196) \\
     & 1 & 1000 & 0.7547 & 0.7393 & 0.6926 & 0.7586 & 0.7420 & 0.7051 \\
     & & & (0.0236) & (0.0242) & (0.0255) & (0.0294) & (0.0286) & (0.0294) \\
     & & 2000 & 0.7657 & 0.7512 & 0.7002 & 0.7741 & 0.7746 & 0.7123 \\
     & & & (0.0166) & (0.0163) & (0.0171) & (0.0197) & (0.0187) & (0.0205) \\
     \bottomrule
\end{tabular}
\end{adjustbox}
}
\end{table}

We perform a sensitivity analysis on the effect of misspecifying the partially linear structure on model performance. The aim of the study is to explore the importance of properly determining the linear and nonlinear parts of the model. We consider the following three scenarios, with all other simulation setups kept unchanged:

\begin{itemize}
\item \textbf{Scenario 1}: $\vz$ is linearly modelled and $\vx$ is nonparametrically modelled,
\item \textbf{Scenario 2}: $Z_1$ is linearly modelled, while $Z_2$ and $\vx$ are nonparametrically modelled,
\item \textbf{Scenario 3}: $\vz$ and $X_1$ are linearly modelled, while the remaining four components of $\vx$ are nonparametrically modelled.
\end{itemize}
Scenario 1 represents the correctly specified model. In Scenario 2, one of the covariates with linear effects is nonlinearly modelled, while the exact opposite happens in Scenario 3. In all scenarios, we obtain the bias and standard deviation of $\widehat{\beta}_1$, and the average and standard deviation of the C-index over 200 simulation runs to evaluate the estimation accuracy and the predictive power, respectively. Analogously, only Case 3 of $g_0$ is involved, and the deep neural network is employed for nonparametric modelling.

It can be inferred from Table \ref{t:webtable10} which summarizes the results that, the model performance under Scenario 1 is merely higher than that under Scenario 2, and is much superior to that under Scenario 3. This points to the conclusion that the correct specification is always supposed to be given the first priority, and in case it is uncertain which covariates linearly affect the response (i.e. the survival time), we can consider inputting all covariates into the deep neural network to achieve relatively better performance.

\bibliographystyle{chicago}
\bibliography{bib}

\begin{thebibliography}{}

\bibitem[\protect\citeauthoryear{Al-Mosawi and Lu}{Al-Mosawi and Lu}{2022}]{al2022efficient}
Al-Mosawi, R. and X.~Lu (2022).
\newblock Efficient estimation of semiparametric varying-coefficient partially linear transformation model with current status data.
\newblock {\em Journal of Statistical Computation and Simulation\/}~{\em 92\/}(2), 416--435.

\bibitem[\protect\citeauthoryear{Anggondowati, Ganti, and Islam}{Anggondowati et~al.}{2020}]{anggondowati2020impact}
Anggondowati, T., A.~K. Ganti, and K.~M. Islam (2020).
\newblock Impact of time-to-treatment on overall survival of non-small cell lung cancer patients—an analysis of the national cancer database.
\newblock {\em Translational lung cancer research\/}~{\em 9\/}(4), 1202.

\bibitem[\protect\citeauthoryear{Austin, Harrell~Jr, and van Klaveren}{Austin et~al.}{2020}]{austin2020graphical}
Austin, P.~C., F.~E. Harrell~Jr, and D.~van Klaveren (2020).
\newblock Graphical calibration curves and the integrated calibration index ({ICI}) for survival models.
\newblock {\em Statistics in Medicine\/}~{\em 39\/}(21), 2714--2742.

\bibitem[\protect\citeauthoryear{Bennett}{Bennett}{1983}]{bennett1983analysis}
Bennett, S. (1983).
\newblock Analysis of survival data by the proportional odds model.
\newblock {\em Statistics in medicine\/}~{\em 2\/}(2), 273--277.

\bibitem[\protect\citeauthoryear{Bickel, Klaassen, Ritov, and Wellner}{Bickel et~al.}{1993}]{bickel1993efficient}
Bickel, P., C.~Klaassen, Y.~Ritov, and J.~Wellner (1993).
\newblock {\em Efficient and adaptive estimation for semiparametric models}, Volume~4.
\newblock Springer.

\bibitem[\protect\citeauthoryear{Breslow}{Breslow}{1972}]{breslow1972discussion}
Breslow, N. (1972).
\newblock Discussion on'regression models and life-tables'{(by DR Cox). J R}.
\newblock {\em Statist. Soc. B\/}~{\em 34}, 216--217.

\bibitem[\protect\citeauthoryear{Chen, Jin, and Ying}{Chen et~al.}{2002}]{chen2002semiparametric}
Chen, K., Z.~Jin, and Z.~Ying (2002).
\newblock Semiparametric analysis of transformation models with censored data.
\newblock {\em Biometrika\/}~{\em 89\/}(3), 659--668.

\bibitem[\protect\citeauthoryear{Collobert, Weston, Bottou, Karlen, Kavukcuoglu, and Kuksa}{Collobert et~al.}{2011}]{collobert2011natural}
Collobert, R., J.~Weston, L.~Bottou, M.~Karlen, K.~Kavukcuoglu, and P.~Kuksa (2011).
\newblock Natural language processing (almost) from scratch.
\newblock {\em Journal of machine learning research\/}~{\em 12}, 2493--2537.

\bibitem[\protect\citeauthoryear{Cox}{Cox}{1972}]{Cox1972}
Cox, D.~R. (1972).
\newblock Regression models and life-tables.
\newblock {\em Journal of the Royal Statistical Society: Series B (Methodological)\/}~{\em 34\/}(2), 187--202.

\bibitem[\protect\citeauthoryear{Cox}{Cox}{1975}]{cox1975partial}
Cox, D.~R. (1975).
\newblock Partial likelihood.
\newblock {\em Biometrika\/}~{\em 62\/}(2), 269--276.

\bibitem[\protect\citeauthoryear{Dabrowska and Doksum}{Dabrowska and Doksum}{1988}]{dabrowska1988estimation}
Dabrowska, D.~M. and K.~A. Doksum (1988).
\newblock Estimation and testing in a two-sample generalized odds-rate model.
\newblock {\em Journal of the american statistical association\/}~{\em 83\/}(403), 744--749.

\bibitem[\protect\citeauthoryear{Du, Wu, Tong, and Zhao}{Du et~al.}{2024}]{du2024deep}
Du, M., Q.~Wu, X.~Tong, and X.~Zhao (2024).
\newblock Deep learning for regression analysis of interval-censored data.
\newblock {\em Electronic Journal of Statistics\/}~{\em 18\/}(2), 4292--4321.

\bibitem[\protect\citeauthoryear{Fine}{Fine}{1999}]{fine1999analysing}
Fine, J. (1999).
\newblock Analysing competing risks data with transformation models.
\newblock {\em Journal of the Royal Statistical Society: Series B (Statistical Methodology)\/}~{\em 61\/}(4), 817--830.

\bibitem[\protect\citeauthoryear{Goodfellow}{Goodfellow}{2016}]{goodfellow2016deep}
Goodfellow, I. (2016).
\newblock Deep learning.

\bibitem[\protect\citeauthoryear{Grigoletto and Akritas}{Grigoletto and Akritas}{1999}]{grigoletto1999analysis}
Grigoletto, M. and M.~G. Akritas (1999).
\newblock Analysis of covariance with incomplete data via semiparametric model transformations.
\newblock {\em Biometrics\/}~{\em 55\/}(4), 1177--1187.

\bibitem[\protect\citeauthoryear{Harrell, Califf, Pryor, Lee, and Rosati}{Harrell et~al.}{1982}]{harrell1982evaluating}
Harrell, F.~E., R.~M. Califf, D.~B. Pryor, K.~L. Lee, and R.~A. Rosati (1982).
\newblock Evaluating the yield of medical tests.
\newblock {\em Jama\/}~{\em 247\/}(18), 2543--2546.

\bibitem[\protect\citeauthoryear{He, Zhang, Ren, and Sun}{He et~al.}{2016}]{he2016deep}
He, K., X.~Zhang, S.~Ren, and J.~Sun (2016).
\newblock Deep residual learning for image recognition.
\newblock {\em Proceedings of the IEEE conference on computer vision and pattern recognition\/}, 770--778.

\bibitem[\protect\citeauthoryear{Heaton, Polson, and Witte}{Heaton et~al.}{2017}]{heaton2017deep}
Heaton, J.~B., N.~G. Polson, and J.~H. Witte (2017).
\newblock Deep learning for finance: deep portfolios.
\newblock {\em Applied Stochastic Models in Business and Industry\/}~{\em 33\/}(1), 3--12.

\bibitem[\protect\citeauthoryear{Katzman, Shaham, Cloninger, Bates, Jiang, and Kluger}{Katzman et~al.}{2018}]{katzman2018deepsurv}
Katzman, J.~L., U.~Shaham, A.~Cloninger, J.~Bates, T.~Jiang, and Y.~Kluger (2018).
\newblock Deepsurv: personalized treatment recommender system using a cox proportional hazards deep neural network.
\newblock {\em BMC medical research methodology\/}~{\em 18}, 1--12.

\bibitem[\protect\citeauthoryear{Kingma and Ba}{Kingma and Ba}{2014}]{adam}
Kingma, D. and J.~Ba (2014).
\newblock Adam: A method for stochastic optimization.
\newblock {\em International Conference on Learning Representations\/}.

\bibitem[\protect\citeauthoryear{Kooperberg, Stone, and Truong}{Kooperberg et~al.}{1995}]{kooperberg1995hazard}
Kooperberg, C., C.~J. Stone, and Y.~K. Truong (1995).
\newblock Hazard regression.
\newblock {\em Journal of the American Statistical Association\/}~{\em 90\/}(429), 78--94.

\bibitem[\protect\citeauthoryear{Kosorok}{Kosorok}{2008}]{kosorok2008}
Kosorok, M.~R. (2008).
\newblock {\em Introduction to Empirical Processes and Semiparametric Inference}.
\newblock Springer New York.

\bibitem[\protect\citeauthoryear{Krizhevsky, Sutskever, and Hinton}{Krizhevsky et~al.}{2012}]{krizhevsky2012imagenet}
Krizhevsky, A., I.~Sutskever, and G.~E. Hinton (2012).
\newblock Imagenet classification with deep convolutional neural networks.
\newblock {\em Advances in neural information processing systems\/}~{\em 25}, 1097--1105.

\bibitem[\protect\citeauthoryear{Kuk and Chen}{Kuk and Chen}{1992}]{kuk1992mixture}
Kuk, A.~Y. and C.-H. Chen (1992).
\newblock A mixture model combining logistic regression with proportional hazards regression.
\newblock {\em Biometrika\/}~{\em 79\/}(3), 531--541.

\bibitem[\protect\citeauthoryear{LeCun, Boser, Denker, Henderson, Howard, Hubbard, and Jackel}{LeCun et~al.}{1989}]{lecun1989backpropagation}
LeCun, Y., B.~Boser, J.~S. Denker, D.~Henderson, R.~E. Howard, W.~Hubbard, and L.~D. Jackel (1989).
\newblock Backpropagation applied to handwritten zip code recognition.
\newblock {\em Neural computation\/}~{\em 1\/}(4), 541--551.

\bibitem[\protect\citeauthoryear{Lee, Zame, Yoon, and Van Der~Schaar}{Lee et~al.}{2018}]{lee2018deephit}
Lee, C., W.~Zame, J.~Yoon, and M.~Van Der~Schaar (2018).
\newblock Deephit: A deep learning approach to survival analysis with competing risks.
\newblock {\em Proceedings of the AAAI conference on artificial intelligence\/}~{\em 32\/}(1), 2314--2321.

\bibitem[\protect\citeauthoryear{Li, Liang, Tong, and Sun}{Li et~al.}{2019}]{li2019estimation}
Li, B., B.~Liang, X.~Tong, and J.~Sun (2019).
\newblock On estimation of partially linear varying-coefficient transformation models with censored data.
\newblock {\em Statistica Sinica\/}~{\em 29\/}(4), 1963--1975.

\bibitem[\protect\citeauthoryear{Lu, Zhang, and Huang}{Lu et~al.}{2007}]{lu2007estimation}
Lu, M., Y.~Zhang, and J.~Huang (2007).
\newblock Estimation of the mean function with panel count data using monotone polynomial splines.
\newblock {\em Biometrika\/}~{\em 94\/}(3), 705--718.

\bibitem[\protect\citeauthoryear{Lu and Ying}{Lu and Ying}{2004}]{lu2004semiparametric}
Lu, W. and Z.~Ying (2004).
\newblock On semiparametric transformation cure models.
\newblock {\em Biometrika\/}~{\em 91\/}(2), 331--343.

\bibitem[\protect\citeauthoryear{Lu and Zhang}{Lu and Zhang}{2010}]{lu2010estimation}
Lu, W. and H.~H. Zhang (2010).
\newblock On estimation of partially linear transformation models.
\newblock {\em Journal of the American Statistical Association\/}~{\em 105\/}(490), 683--691.

\bibitem[\protect\citeauthoryear{Ma and Kosorok}{Ma and Kosorok}{2005}]{ma2005penalized}
Ma, S. and M.~R. Kosorok (2005).
\newblock Penalized log-likelihood estimation for partly linear transformation models with current status data.
\newblock {\em The Annals of Statistics\/}~{\em 33\/}(5), 2256--2290.

\bibitem[\protect\citeauthoryear{Norman, Li, Jiang, and Chen}{Norman et~al.}{2024}]{norman2024deepaft}
Norman, P.~A., W.~Li, W.~Jiang, and B.~E. Chen (2024).
\newblock deepaft: A nonlinear accelerated failure time model with artificial neural network.
\newblock {\em Statistics in Medicine\/}~{\em 43}, 3689--3701.

\bibitem[\protect\citeauthoryear{Ohn and Kim}{Ohn and Kim}{2022}]{ohn2022nonconvex}
Ohn, I. and Y.~Kim (2022).
\newblock Nonconvex sparse regularization for deep neural networks and its optimality.
\newblock {\em Neural computation\/}~{\em 34\/}(2), 476--517.

\bibitem[\protect\citeauthoryear{Paszke, Gross, Massa, Lerer, Bradbury, Chanan, et~al.}{Paszke et~al.}{2019}]{NEURIPS2019_9015}
Paszke, A., S.~Gross, F.~Massa, A.~Lerer, J.~Bradbury, G.~Chanan, et~al. (2019).
\newblock Pytorch: An imperative style, high-performance deep learning library.
\newblock In H.~Wallach, H.~Larochelle, A.~Beygelzimer, F.~d\textquotesingle Alch\'{e}-Buc, E.~Fox, and R.~Garnett (Eds.), {\em Advances in Neural Information Processing Systems 32}, pp.\  8024--8035. Curran Associates, Inc.

\bibitem[\protect\citeauthoryear{Schmidt-Hieber}{Schmidt-Hieber}{2020}]{SchmidtHieber2020}
Schmidt-Hieber, J. (2020).
\newblock Nonparametric regression using deep neural networks with {ReLU} activation function.
\newblock {\em The Annals of Statistics\/}~{\em 48\/}(4), 1875–1897.

\bibitem[\protect\citeauthoryear{Schumaker}{Schumaker}{2007}]{schumaker_2007}
Schumaker, L. (2007).
\newblock {\em Spline Functions: Basic Theory\/} (3 ed.).
\newblock Cambridge: Cambridge University Press.

\bibitem[\protect\citeauthoryear{Shen and Wong}{Shen and Wong}{1994}]{shen1994convergence}
Shen, X. and W.~H. Wong (1994).
\newblock Convergence rate of sieve estimates.
\newblock {\em The Annals of Statistics\/}, 580--615.

\bibitem[\protect\citeauthoryear{Srivastava, Hinton, Krizhevsky, Sutskever, and Salakhutdinov}{Srivastava et~al.}{2014}]{srivastava2014dropout}
Srivastava, N., G.~Hinton, A.~Krizhevsky, I.~Sutskever, and R.~Salakhutdinov (2014).
\newblock Dropout: a simple way to prevent neural networks from overfitting.
\newblock {\em The journal of machine learning research\/}~{\em 15\/}(1), 1929--1958.

\bibitem[\protect\citeauthoryear{Stone}{Stone}{1985}]{stone1985additive}
Stone, C.~J. (1985).
\newblock Additive regression and other nonparametric models.
\newblock {\em The annals of Statistics\/}~{\em 13\/}(2), 689--705.

\bibitem[\protect\citeauthoryear{Su, Liu, Yin, Huang, and Zhao}{Su et~al.}{2024}]{su2024deep}
Su, W., K.-Y. Liu, G.~Yin, J.~Huang, and X.~Zhao (2024).
\newblock Deep nonparametric inference for conditional hazard function.
\newblock {\em arXiv preprint arXiv:2410.18021\/}.

\bibitem[\protect\citeauthoryear{Sun, Kang, Haridas, Mayne, Potter, Yang, Christiani, and Li}{Sun et~al.}{2024}]{sun2024penalized}
Sun, Y., J.~Kang, C.~Haridas, N.~Mayne, A.~Potter, C.-F. Yang, D.~C. Christiani, and Y.~Li (2024).
\newblock Penalized deep partially linear cox models with application to ct scans of lung cancer patients.
\newblock {\em Biometrics\/}~{\em 80\/}(1), ujad024.

\bibitem[\protect\citeauthoryear{Tsybakov}{Tsybakov}{2009}]{tsybakov2009nonparametric}
Tsybakov, A.~B. (2009).
\newblock Nonparametric estimators.
\newblock {\em Introduction to Nonparametric Estimation\/}, 1--76.

\bibitem[\protect\citeauthoryear{Van~der Vaart}{Van~der Vaart}{2000}]{van2000asymptotic}
Van~der Vaart, A.~W. (2000).
\newblock {\em Asymptotic Statistics}.
\newblock Cambridge university press.

\bibitem[\protect\citeauthoryear{Van Der~Vaart and Wellner}{Van Der~Vaart and Wellner}{1996}]{van1996weak}
Van Der~Vaart, A.~W. and J.~A. Wellner (1996).
\newblock {\em Weak Convergence and Empirical Processes}.
\newblock Springer.

\bibitem[\protect\citeauthoryear{Vaswani, Shazeer, Parmar, Uszkoreit, Jones, Gomez, Kaiser, and Polosukhin}{Vaswani et~al.}{2017}]{vaswani2017attention}
Vaswani, A., N.~Shazeer, N.~Parmar, J.~Uszkoreit, L.~Jones, A.~N. Gomez, {\L}.~Kaiser, and I.~Polosukhin (2017).
\newblock Attention is all you need.
\newblock {\em Advances in neural information processing systems\/}~{\em 30}.

\bibitem[\protect\citeauthoryear{Wang, Wang, Sun, Cao, and Zhao}{Wang et~al.}{2022}]{wang2022evaluation}
Wang, Q., S.~Wang, Z.~Sun, M.~Cao, and X.~Zhao (2022).
\newblock Evaluation of log odds of positive lymph nodes in predicting the survival of patients with non-small cell lung cancer treated with neoadjuvant therapy and surgery: a seer cohort-based study.
\newblock {\em BMC cancer\/}~{\em 22\/}(1), 801.

\bibitem[\protect\citeauthoryear{Wei}{Wei}{1992}]{wei1992accelerated}
Wei, L.-J. (1992).
\newblock The accelerated failure time model: a useful alternative to the cox regression model in survival analysis.
\newblock {\em Statistics in medicine\/}~{\em 11\/}(14-15), 1871--1879.

\bibitem[\protect\citeauthoryear{Wu, Tong, and Zhao}{Wu et~al.}{2024}]{wu2024deep}
Wu, Q., X.~Tong, and X.~Zhao (2024).
\newblock Deep partially linear cox model for current status data.
\newblock {\em Biometrics\/}~{\em 80\/}(2), ujae024.

\bibitem[\protect\citeauthoryear{Wu, Qiao, Wu, Yu, Zheng, Liu, Zhang, and Wang}{Wu et~al.}{2023}]{wu2023neural}
Wu, R., J.~Qiao, M.~Wu, W.~Yu, M.~Zheng, T.~Liu, T.~Zhang, and W.~Wang (2023).
\newblock Neural frailty machine: Beyond proportional hazard assumption in neural survival regressions.
\newblock {\em Advances in Neural Information Processing Systems\/}~{\em 36}, 5569--5597.

\bibitem[\protect\citeauthoryear{Xie and Yu}{Xie and Yu}{2021}]{xie2021promotion}
Xie, Y. and Z.~Yu (2021).
\newblock Promotion time cure rate model with a neural network estimated nonparametric component.
\newblock {\em Statistics in Medicine\/}~{\em 40\/}(15), 3516--3532.

\bibitem[\protect\citeauthoryear{Yarotsky}{Yarotsky}{2017}]{yarotsky2017error}
Yarotsky, D. (2017).
\newblock Error bounds for approximations with deep relu networks.
\newblock {\em Neural networks\/}~{\em 94}, 103--114.

\bibitem[\protect\citeauthoryear{Zeleniuch-Jacquotte, Shore, Koenig, Akhmedkhanov, Afanasyeva, Kato, Kim, Rinaldi, Kaaks, and Toniolo}{Zeleniuch-Jacquotte et~al.}{2004}]{zeleniuch2004postmenopausal}
Zeleniuch-Jacquotte, A., R.~Shore, K.~Koenig, A.~Akhmedkhanov, Y.~Afanasyeva, I.~Kato, M.~Kim, S.~Rinaldi, R.~Kaaks, and P.~Toniolo (2004).
\newblock Postmenopausal levels of oestrogen, androgen, and shbg and breast cancer: long-term results of a prospective study.
\newblock {\em British journal of cancer\/}~{\em 90\/}(1), 153--159.

\bibitem[\protect\citeauthoryear{Zeng and Lin}{Zeng and Lin}{2007}]{zeng2007semiparametric}
Zeng, D. and D.~Lin (2007).
\newblock Semiparametric transformation models with random effects for recurrent events.
\newblock {\em Journal of the American Statistical Association\/}~{\em 102\/}(477), 167--180.

\bibitem[\protect\citeauthoryear{Zeng, Mao, and Lin}{Zeng et~al.}{2016}]{zeng2016maximum}
Zeng, D., L.~Mao, and D.~Lin (2016).
\newblock Maximum likelihood estimation for semiparametric transformation models with interval-censored data.
\newblock {\em Biometrika\/}~{\em 103\/}(2), 253--271.

\bibitem[\protect\citeauthoryear{Zeng, Zhang, Chen, and Ding}{Zeng et~al.}{2025}]{zeng2025tdcoxsnn}
Zeng, L., J.~Zhang, W.~Chen, and Y.~Ding (2025).
\newblock tdcoxsnn: Time-dependent cox survival neural network for continuous-time dynamic prediction.
\newblock {\em Journal of the Royal Statistical Society Series C: Applied Statistics\/}~{\em 74\/}(1), 187--203.

\bibitem[\protect\citeauthoryear{Zhang, Tong, Zhang, Wang, and Sun}{Zhang et~al.}{2013}]{zhang2013efficient}
Zhang, B., X.~Tong, J.~Zhang, C.~Wang, and J.~Sun (2013).
\newblock Efficient estimation for linear transformation models with current status data.
\newblock {\em Communications in Statistics-Theory and Methods\/}~{\em 42\/}(17), 3191--3203.

\bibitem[\protect\citeauthoryear{Zhang and Zhang}{Zhang and Zhang}{2023}]{zhang2023prognostic}
Zhang, J. and J.~Zhang (2023).
\newblock Prognostic factors and survival prediction of resected non-small cell lung cancer with ipsilateral pulmonary metastases: a study based on the surveillance, epidemiology, and end results (seer) database.
\newblock {\em BMC Pulmonary Medicine\/}~{\em 23\/}(1), 413.

\bibitem[\protect\citeauthoryear{Zhong, Mueller, and Wang}{Zhong et~al.}{2022}]{Zhong2022}
Zhong, Q., J.~Mueller, and J.-L. Wang (2022).
\newblock Deep learning for the partially linear cox model.
\newblock {\em The Annals of Statistics\/}~{\em 50\/}(3), 1348--1375.

\end{thebibliography}

\end{document}